\documentclass[a4paper, 12pt, twoside]{article}

\usepackage[a4paper, margin=1in]{geometry}
\usepackage[latin1]{inputenc}
\usepackage[T1]{fontenc}
\usepackage[english]{babel}
\usepackage{textcomp} 
\usepackage{graphicx}
\DeclareGraphicsExtensions{.jpg,.mps,.pdf,.png,.gif}
\usepackage[french]{varioref} 
\usepackage[]{amsmath,amssymb,amsfonts}

\usepackage{float}
\usepackage{bbm}
\usepackage{amsthm}
\usepackage{dsfont}
\usepackage{color}
\usepackage{natbib}
\usepackage{pict2e}
\usepackage[hyperindex,breaklinks]{hyperref}

\newcommand{\Ben}{\begin{enumerate}}
\newcommand{\Een}{\end{enumerate}}
\newcommand{\Bit}{\begin{itemize}}
\newcommand{\Eit}{\end{itemize}}
\newcommand{\Beq}{\begin{equation}}
\newcommand{\Eeq}{\end{equation}}
\newcommand{\Ba}{\begin{align*}}
\newcommand{\Ea}{\end{align*}}
\newcommand{\Mb}{\mathbf}

\newtheorem{Th}{Theorem}

\newtheorem{Prop}{Proposition}

\newtheorem{Rq}{Remark} 
\newtheorem{Corr}{Corollary}
\newtheorem{Def}{Definition}

\title{Spatial risk measures and rate of spatial diversification}
\date{June 28, 2019}
\begin{document}
\author{
Erwan Koch\footnote{EPFL, Chair of Statistics STAT, EPFL-SB-MATH-STAT, MA B1 433 (B\^atiment MA), Station 8,
1015 Lausanne, Switzerland. Email: erwan.koch@epfl.ch} 
}
\maketitle

\begin{abstract}
An accurate assessment of the risk of extreme environmental events is of great importance for populations, authorities and the banking/insurance/reinsurance industry. \cite{koch2017spatial} introduced a notion of spatial risk measure and a corresponding set of axioms which are well suited to analyze the risk due to events having a spatial extent, precisely such as environmental phenomena. The axiom of asymptotic spatial homogeneity is of particular interest since it allows one to quantify the rate of spatial diversification when the region under consideration becomes large. In this paper, we first investigate the general concepts of spatial risk measures and corresponding axioms further and thoroughly explain the usefulness of this theory for both actuarial science and practice. Second, in the case of a general cost field, we give sufficient conditions such that spatial risk measures associated with expectation, variance, Value-at-Risk as well as expected shortfall and induced by this cost field satisfy the axioms of asymptotic spatial homogeneity of order $0$, $-2$, $-1$ and $-1$, respectively. Last but not least, in the case where the cost field is a function of a max-stable random field, we provide conditions on both the function and the max-stable field ensuring the latter properties. Max-stable random fields are relevant when assessing the risk of extreme events since they appear as a natural extension of multivariate extreme-value theory to the level of random fields. Overall, this paper improves our understanding of spatial risk measures as well as of their properties with respect to the space variable and generalizes many results obtained in \cite{koch2017spatial}.

\medskip

\noindent \textbf{Key words}: Central limit theorem; Insurance; Max-stable random fields; Rate of spatial diversification; Reinsurance; Risk management; Risk theory; Spatial dependence; Spatial risk measures and corresponding axiomatic approach.
\end{abstract}

\section{Introduction}

Hurricane Irma, which affected many Caribbean islands and parts of Florida in September 2017 caused at least 134 deaths and
catastrophic damage exceeding 64.8 billion USD in value. Such an example shows the prime importance for civil authorities and for the insurance\footnote{Throughout the paper, insurance also refers to reinsurance.} industry of the accurate assessment of the risk of natural disasters, particularly as, in a climate change context, certain types of extreme events become more and more frequent \citep[e.g.,][]{SwissRe}.

Motivated by the spatial feature of natural disasters, \cite{koch2017spatial} introduced a new notion of spatial risk measure, which makes explicit the contribution of the space and enables one to account for at least part of the spatial dependence in the risk measurement. He also introduced a set of axioms describing how the risk is expected to evolve with respect to the space variable, at least under some conditions. 
These notions constitute relevant tools for risk assessment. For instance, the knowledge of the order of asymptotic spatial homogeneity allows the quantification of the rate of spatial diversification. Hence, they may be appealing for the banking/insurance industry. It should be highlighted that the literature about risk measures in a spatial context is very limited. To the best of our knowledge, the paper by \cite{koch2017spatial} constitutes the first attempt to establish a theory about risk measures in a spatial context where the risks spread over a continuous geographical region. 

In the following, the spatial risk measure associated with a classical risk measure $\Pi$ and induced by a cost random field $C$ (e.g., modelling the cost due to damage caused by a natural disaster) consists in the function of space arising from the application of $\Pi$ to the normalized integral of $C$ on various geographical areas. The contribution of this paper is threefold. First, we further explore the notions of spatial risk measure and corresponding axioms introduced in \cite{koch2017spatial}. Among others, we show that, for a given region, the distribution of the normalized spatially aggregated loss is entirely determined by the finite-dimensional distributions of the cost field, and propose alternative definitions of the concepts developed in \cite{koch2017spatial}. Additionally, we deeply explain why this whole theory about spatial risk measures is fruitful for both actuarial science and practice; e.g., we show that considering the risk related to the normalized loss does not prevent our theory from being successful for the study of the risk related to the non-normalized loss. We also point out how it can be used by insurance companies to tackle concrete issues. Second, in the case of a general cost field, we give sufficient conditions such that spatial risk measures associated with expectation, variance, Value-at-Risk (VaR) as well as expected shortfall (ES) and induced by this cost field satisfy the axiom of asymptotic spatial homogeneity of order $0$, $-2$, $-1$ and $-1$, respectively. Last but not least, we focus on the case where the cost field is a function of a max-stable random field. We provide sufficient conditions on both the function and the max-stable field such that spatial risk measures associated with expectation, variance, VaR as well as ES and induced by the resulting cost field satisfy the axiom of asymptotic spatial homogeneity of order $0$, $-2$, $-1$ and $-1$, respectively. Max-stable random fields naturally appear when one is interested in extreme events having a spatial extent since they constitute an extension of multivariate extreme-value theory to the level of random fields \citep[in the case of stochastic processes, see, e.g.,][]{haan1984spectral, de2007extreme}. They are particularly well suited to model the temporal maxima of a given variable (for instance a meteorological variable) at all points in space since they arise as the pointwise maxima taken over an infinite number of appropriately rescaled independent and identically distributed random fields. On the whole, this study improves our comprehension of spatial risk measures as well as of their properties with respect to the space variable and generalizes many results by \cite{koch2017spatial}.

The remainder of the paper is organized as follows. In Section \ref{Sec_Reminder_Spatial_Risk_Measures}, we recall and further study the notion of spatial risk measure and the corresponding set of axioms introduced in \cite{koch2017spatial}. Furthermore, we thoroughly demonstrate their usefulness for both actuarial science and practice. Then, we introduce some concepts about mixing and central limit theorems for random fields. Finally, we provide some insights about max-stable random fields. Then, Section \ref{Sec_Axioms} presents our results relating to the properties of some spatial risk measures. Finally, Section \ref{Sec_Conclusion} contains a short summary as well as some perspectives. 

Throughout the paper, $(\Omega, \mathcal{F}, \mathbb{P})$ is an adequate probability space and $\overset{d}{=}$ and $\overset{d}{\rightarrow}$ designate equality and convergence in distribution, respectively. In the case of random fields, distribution has to be understood as the set of all finite-dimensional distributions. Finally, we denote by $\nu$ the Lebesgue measure.

\section{Spatial risk measures and other concepts}
\label{Sec_Reminder_Spatial_Risk_Measures}

\subsection{Spatial risk measures and corresponding axioms}

First, we describe the setting required for a proper definition of spatial risk measures. Let $\mathcal{A}$ be the set of all compact subsets of $\mathbb{R}^2$ with a positive Lebesgue measure and $\mathcal{A}_c$ be the set of all convex elements of $\mathcal{A}$. Denote by $\mathcal{C}$ the set of all real-valued and measurable\footnote{Throughout the paper, when applied to random fields, the adjective ``measurable'' means ``jointly measurable''.} random fields on $\mathbb{R}^2$ having almost surely (a.s.)\footnote{Unless otherwise stated, by a.s., we mean $\mathds{P}$-a.s.} locally integrable sample paths. Let $\mathcal{P}$ be the family of all possible distributions of random fields belonging to $\mathcal{C}$. Each random field represents the economic or insured cost caused by the events belonging to specified classes and occurring during a given time period, say $[0,T_L]$. In the following, $T_L$ is considered as fixed and does not appear anymore for the sake of notational parsimony. 
Each class of events (e.g., European windstorms or hurricanes) will be referred to as a hazard in the following. Let $\mathcal{L}$ be the set of all real-valued random variables defined on $(\Omega, \mathcal{F}, \mathbb{P})$. A risk measure typically will be some function $\Pi: \mathcal{L} \to \mathbb{R}$. This kind of risk measure will be referred to as a classical risk measure in the following. A classical risk measure $\Pi$ is termed law-invariant if, for all $\tilde{X} \in \mathcal{L}$, $\Pi ( \tilde{X} )$ only depends on the distribution of $\tilde{X}$.

We first remind the reader of the definition of the normalized spatially aggregated loss, which enables one to disentangle the contribution of the space and the contribution of the hazards and underpins our definition of spatial risk measure.

\begin{Def}[Normalized spatially aggregated loss as a function of the distribution of the cost field]
\label{Def_normalized_Loss}
For $A \in \mathcal{A}$ and $P \in \mathcal{P}$, the normalized spatially aggregated loss is defined by
\begin{equation}
\label{Chap_RiskMeasure_Eq_Loss_Normalized}
L_N(A,P)= \dfrac{1}{\nu(A)} \displaystyle\int_{A} C_P(\Mb{x}) \  \nu(\mathrm{d}\Mb{x}),
\end{equation}
where the random field $\{ C_P(\Mb{x}) \}_{\Mb{x} \in \mathbb{R}^2}$ belongs to $\mathcal{C}$ and has distribution $P$.
\end{Def}
The quantity 
\Beq
\label{Eq_Loss_Function_Dist}
L(A, P)=\int_{A} C_P(\Mb{x}) \  \nu(\mathrm{d}\Mb{x})
\Eeq
corresponds to the total economic or insured loss over region $A$ due to specified hazards. For technical reasons and to favour a more intuitive understanding, we base our definition of spatial risk measures on $L_N(A,P)$, which is the loss per surface unit and can be interpreted, in a discrete setting\footnote{See Section \ref{Subsec_Concrete_App_Ins_Reins_Industry}.} and in an insurance context, as the mean loss per insurance policy. Among other advantages, this normalization enables a fair comparison of the risks related to regions having different sizes.

Since the field $C_P$ is measurable, $L(A,P)$ and $L_N(A,P)$ are well-defined random variables. Moreover, they are a.s. finite as A is compact and $C_P$ has a.s. locally integrable sample paths. The following proposition gives a sufficient condition for a random field to have a.s. locally integrable sample paths.

\begin{Prop}
\label{Eq_Sufficient_Condition_Locally_Integrable_Sample_Paths}
Let $d \geq 1$ and $\{ Q(\Mb{x}) \}_{\Mb{x} \in \mathbb{R}^d}$ be a measurable random field. If the function 
$$\begin{array}{ccccc}
E &:& \mathbb{R}^d & \to & \mathbb{R} \\
& & \Mb{x} & \mapsto & \mathbb{E}[|Q(\Mb{x})|]
\end{array}
$$
is locally integrable, then $Q$ has a.s. locally integrable sample paths. 
\end{Prop}

\begin{proof}
Let $A$ be a compact subset of $\mathbb{R}^d$. First, since $Q$ is measurable, $\int_{A} |Q(\Mb{x})| \ \nu(\mathrm{d}\Mb{x})$ is a well-defined random variable. By Fubini's theorem, we have 
$$ \mathbb{E} \left[ \int_{A} |Q(\Mb{x})| \ \nu(\mathrm{d}\Mb{x}) \right]=\int_{A} \mathbb{E}[|Q(\Mb{x})|] \ \nu(\mathrm{d}\Mb{x})< \infty,$$
which necessarily implies that 
$$\int_{A} |Q(\Mb{x})| \ \nu(\mathrm{d}\Mb{x}) < \infty \ \  \mbox{ a.s.}$$ 
Since this is true for all $A$ being a compact subset of $\mathbb{R}^d$, we obtain the result.
\end{proof}

We now recall the notion of spatial risk measure introduced by \cite{koch2017spatial}, which makes explicit the contribution of the space in the risk measurement.
\begin{Def}[Spatial risk measure as a function of the distribution of the cost field]
\label{Def_Spatial_Risk_Measure}
A spatial risk measure is a function $\mathcal{R}_{\Pi}$ that assigns a real number to any region $A \in \mathcal{A}$ and distribution $P \in \mathcal{P}$:
$$\begin{array}{ccccc}
\mathcal{R}_{\Pi} & : & \mathcal{A} \times \mathcal{P}  & \to & \mathbb{R} \\
  &   & (A,P) & \mapsto & \Pi ( L_N(A,P) ),
\end{array}$$
where $\Pi$ is a classical and law-invariant risk measure and $L_N(A,P)$ is defined in \eqref{Chap_RiskMeasure_Eq_Loss_Normalized}.
\end{Def}
This extends the notion of classical risk measure to the spatial and infinite-dimensional setting as we now have a function of both the space and the distribution of a random field (or directly  a random field, see below) instead of a function of a unique real-valued random variable. Note that law-invariance of $\Pi$ is necessary for spatial risk measures to be defined in this way; see below for more details. For a given $\Pi$ and a fixed $P \in \mathcal{P}$, the quantity $\mathcal{R}_{\Pi}(\cdot,P)$ is referred to as the spatial risk measure associated with $\Pi$ and induced by $P$. A nice feature is that, for many useful classical risk measures $\Pi$ such as, e.g., variance, VaR and ES, this notion of spatial risk measure allows one to take (at least) part of the spatial dependence structure of the field $C_P$ into account. We could define spatial risk measures in the same way but using the non-normalized spatially aggregated loss; this is not what we do for reasons explained above and in Remark \ref{Rq_Subadd_Nonnormalized} below.

Now, we remind the reader of the set of axioms for spatial risk measures developed in \cite{koch2017spatial}. It concerns the spatial risk measures properties with respect to the space and not to the cost distribution, the latter being considered as given by the problem at hand. For any $A \in \mathcal{A}$, let $\Mb{b}_A$ denote its barycenter.
\begin{Def}[Set of axioms for spatial risk measures induced by a distribution]
\label{Chapriskmeasures_Def_Axiomatic}
Let $\Pi$ be a classical and law-invariant risk measure. For a fixed $P \in \mathcal{P}$, we define the following axioms for the spatial risk measure associated with $\Pi$ and induced by $P$, $\mathcal{R}_{\Pi}(\cdot, P)$:
\medskip
\Ben
\item \textbf{Spatial invariance under translation:} \\ for all $\Mb{v} \in \mathbb{R}^2$ and $A \in \mathcal{A},  \ \mathcal{R}_{\Pi}(A+\Mb{v}, P)=\mathcal{R}_{\Pi}(A,P)$, where $A+\Mb{v}$ denotes the region $A$ translated by the vector $\Mb{v}$.
\item \textbf{Spatial sub-additivity:} \\
for all $A_1, A_2 \in \mathcal{A},\  \mathcal{R}_{\Pi}(A_1 \cup A_2, P) \leq \min \{ \mathcal{R}_{\Pi}(A_1, P),\mathcal{R}_{\Pi}(A_2, P) \}$.
\item \textbf{Asymptotic spatial homogeneity of order $\boldsymbol{- \gamma}\boldsymbol{, \gamma \geq 0}$:} \\ for all $A \in \mathcal{A}_c$,
$$
\mathcal{R}_{\Pi}(\lambda A, P) \underset{\lambda \to \infty}{=} K_1(A,P)+\dfrac{K_2(A,P)}{\lambda^{\gamma}} + o\left(\frac{1}{\lambda^{\gamma}}\right),
$$
where $\lambda A$ is the area obtained by applying to $A$ a homothety with center $\Mb{b}_A$ and ratio $\lambda >0$, and 
$K_1(\cdot, P): \mathcal{A}_c  \to \mathbb{R}$, $K_2(\cdot, P): \mathcal{A}_c  \to \mathbb{R} \backslash \{ 0 \}$ are functions depending on $P$.
\Een
\end{Def}
It is also reasonable to introduce the axiom of \textbf{spatial anti-monotonicity:}
for all $A_1, A_2 \in \mathcal{A}$, $A_1 \subset A_2 \Rightarrow \mathcal{R}_{\Pi}(A_2, P) \leq \mathcal{R}_{\Pi}(A_1, P)$. The latter is equivalent to the axiom of spatial sub-additivity. These axioms appear natural and make sense at least under some conditions on the cost field $C_P$ (e.g., stationarity\footnote{Throughout the paper, stationarity refers to strict stationarity.} in the case of spatial invariance under translation and spatial sub-additivity) and for some classical risk measures $\Pi$. The axiom of spatial sub-additivity indicates spatial diversification. If it is satisfied with strict inequality, an insurance company would be well advised to underwrite policies in both regions $A_1$ and $A_2$ instead of only one of them. This axiom involves the minimum operator because the concept of spatial risk measure is based on the \textit{normalized} spatially aggregated loss; using the summation operator instead would not provide information about spatial diversification. On the other hand, if spatial risk measures were defined using the \textit{non-normalized} loss, then summation would make sense; see Remark \ref{Rq_Subadd_Nonnormalized} below for more details. Originally, \cite{koch2017spatial} used the term ``sub-additivity'', among other reasons, by analogy with the axiom of sub-additivity by \cite{artzner1999coherent}, which also conveys an idea of diversification. The axiom of asymptotic spatial homogeneity of order $-\gamma$ quantifies the rate of spatial diversification when the region becomes large. Consequently, determining the value of $\gamma$ is of interest for the insurance industry; see Section \ref{Subsec_Concrete_App_Ins_Reins_Industry} for further details.

The axioms of spatial invariance under translation and spatial sub-additivity a priori make sense only if the cost field satisfies at least some kind of stationarity. If an insurance company covers a region $A_1$ which is much less risky than a region $A_2$, it is very unlikely that the company reduces its risk by covering $A_1 \cup A_2$. For a given hazard (e.g., hurricanes), the cost resulting from a single specific event (e.g., a particular hurricane) generally varies across space, making any particular realization of the cost field spatially inhomogeneous. Nevertheless, the cost field (and not one realization of it) related to this hazard can be stationary, or, at least, piecewise stationary; see immediately below.

In concrete actuarial applications, the cost field (for a given hazard) is often non-stationary over the entire region covered by the insurance company, unless it is a very small area. In many cases, however, it can reasonably be considered as locally stationary; see, e.g., \cite{dahlhaus2012locally} for an excellent review about locally stationary processes, and \cite{eckley2010locally} as well as \cite{anderes2011local} for papers dealing with local non-stationarity in the case of random fields. Locally stationary processes can be well approximated by piecewise stationary processes \citep[e.g.,][Section 2.2]{ombao2001automatic} and, assuming this to be also true for random fields, we can reasonably consider the cost field to be stationary over sub-regions, at least in most cases. In the latter, the axioms of spatial invariance under translation and spatial sub-additivity make sense separately on each sub-region over which the field is stationary. Let $S_{\mathrm{ub}}$ be such a sub-region (a subset of $\mathbb{R}^2$) and $\mathcal{S_{\mathrm{ub}}}$ be the set of all compact subsets of $S_{\mathrm{ub}}$ with a positive Lebesgue measure. The axiom of spatial invariance under translation becomes:
for all $\Mb{v} \in \mathbb{R}^2$ and $A \in \mathcal{S_{\mathrm{ub}}}$ such that $A+\Mb{v} \in \mathcal{S_{\mathrm{ub}}}$,  $\ \mathcal{R}_{\Pi}(A+\Mb{v}, P)=\mathcal{R}_{\Pi}(A,P)$; spatial sub-additivity is now written:
for all $A_1, A_2 \in \mathcal{S_{\mathrm{ub}}},\  \mathcal{R}_{\Pi}(A_1 \cup A_2, P) \leq \min \{ \mathcal{R}_{\Pi}(A_1, P),\mathcal{R}_{\Pi}(A_2, P) \}$.

Of course, the fact that the axioms of Definition \ref{Chapriskmeasures_Def_Axiomatic} are satisfied depends on both the classical risk measure $\Pi$ and the cost field $C_P$. It may be interesting to determine for which classical risk measures the axioms are satisfied for the broadest class of cost fields. These classical risk measures could be considered as ``adapted'' to the spatial context.

\begin{Rq}
Although the concept of spatial risk measure and related axioms naturally apply in an insurance context (see Section \ref{Subsec_Concrete_App_Ins_Reins_Industry} for further details), they can also be used in the banking industry and on financial markets. A potential application is the assessment of the risk related to event-linked securities such as CAT bonds. Furthermore, they can be used for a wider class of risks than those linked with damage due to environmental events. These concepts are actually insightful as soon as the risks spread over a geographical region. One might think, e.g., about the loss in value of real estate due to adverse economic conditions. 
\end{Rq}

We close this section by deeply commenting on the previous concepts and giving slightly modified and more natural versions of previous definitions. First, we need the following useful result.
\begin{Th}
\label{Th_Characterization_LN_Finite_Dimensional_Distributions_CP}
Let $d \geq 1$ and $\{ H(\Mb{x}) \}_{\Mb{x} \in \mathbb{R}^d}$ be a measurable random field having a.s. locally integrable sample paths. Moreover, let $A$ be a compact subset of $\mathbb{R}^d$ with positive Lebesgue measure. Then the distribution of 
$$ L_N(A, H)=\frac{1}{\nu(A)} \int_A H(\Mb{x}) \ \nu(\mathrm{d}\Mb{x})$$ 
only depends on $A$ and the finite-dimensional distributions of $H$.
\end{Th}

\begin{proof}
The proof is partly inspired from the proof of Theorem 11.4.1 in \cite{samorodnitsky1994stable}. We assume that the random field $H$ is defined on the probability space $(\Omega, \mathcal{F}, \mathbb{P})$. For a fixed $\omega \in \Omega$, we denote by $H_{\omega}$ the corresponding realization of $H$ on $\mathbb{R}^d$ and by $H_{\omega}(\Mb{x})$ the realization of $H$ at location $\Mb{x}$. By definition, we have, for almost every $\omega \in \Omega$, that
\Beq
\label{Eq_Def_Normalized_Loss_omega}
L_N(A,H_{\omega})= \frac{1}{\nu(A)} \int_A H_{\omega}(\Mb{x}) \ \nu(\mathrm{d}\Mb{x}).
\Eeq
Now, let $(\Omega_1, \mathcal{F}_1, \mathbb{P}_1)$ be a probability space different from the probability space $(\Omega, \mathcal{F}, \mathbb{P})$. Let $\Mb{U}$ be a random vector defined on $(\Omega_1, \mathcal{F}_1, \mathbb{P}_1)$ and following the uniform distribution on $A$, with density $f_{\Mb{U}}(\Mb{x})=\mathbb{I}_{ \{ \Mb{x} \in A \} }/\nu(A), \Mb{x} \in \mathbb{R}^d$. From \eqref{Eq_Def_Normalized_Loss_omega}, it directly follows that, for almost every $\omega \in \Omega$,
\Beq
\label{Eq_Normalized_Loss_Density_Uniform}
L_N(A,H_{\omega})= \int_{\mathbb{R}^d} H_{\omega}(\Mb{x}) f_{\Mb{U}}(\Mb{x}) \ \nu(\mathrm{d}\Mb{x}).
\Eeq
Let us denote by $\mathbb{E}_1$ the expectation with respect to the probability measure $\mathbb{P}_1$.
We have
$$ \mathbb{E}_1[H_{\omega}(\Mb{U})]= \int_{\mathbb{R}^d} H_{\omega}(\Mb{x}) f_{\Mb{U}}(\Mb{x}) \ \nu(\mathrm{d}\Mb{x}),$$
giving, using \eqref{Eq_Normalized_Loss_Density_Uniform},
$$L_N(A,H_{\omega})=\mathbb{E}_1[H_{\omega}(\Mb{U})].$$
Now, let $\Mb{U}_1, \dots, \Mb{U}_n$ be independent replications of $\Mb{U}$ (which are independent of the random field $H$). The strong law of large numbers gives that, for almost every $\omega \in \Omega$,
\Beq
\label{Eq_Proof_Integral_Finite_Dimensional_Distributions}
L_N(A,H_{\omega})=\lim_{n \to \infty} \frac{1}{n} \sum_{i=1}^n H_{\omega}(\Mb{U}_i) \ \mathbb{P}_1\mbox{-a.s.}
\Eeq
Therefore, using Fubini's theorem, we deduce that, for $\mathbb{P}_1$-almost every $\omega_1 \in \Omega_1$,
\Beq
\label{Eq_Final_Equation_Proof_Integral_Finite_Dimensional_Distributions}
L_N(A,H_{\omega})=\lim_{n \to \infty} \frac{1}{n} \sum_{i=1}^n H_{\omega}(\Mb{U}_i(\omega_1)) \ \mathbb{P}\mbox{-a.s.}
\Eeq
Now, we choose $\omega_0 \in \Omega_1$ such that the (non-random) sequence $(\Mb{U}_1(\omega_0), \Mb{U}_2(\omega_0), \dots)$ satisfies \eqref{Eq_Final_Equation_Proof_Integral_Finite_Dimensional_Distributions}.
We obtain
\Beq
\label{Eq_Final_Final_Equation_Proof_Integral_Finite_Dimensional_Distributions}
L_N(A,H_{\omega})=\lim_{n \to \infty} \frac{1}{n} \sum_{i=1}^n H_{\omega}(\Mb{U}_i(\omega_0)) \ \mathbb{P}\mbox{-a.s.}
\Eeq
Equation \eqref{Eq_Final_Final_Equation_Proof_Integral_Finite_Dimensional_Distributions} says that the distribution of $L_N(A,H)$ is determined by the finite-dimensional distributions at the points belonging to the set $\{ \Mb{U}_i(\omega_0): i \in \mathbb{N} \}$.
This yields the result.
\end{proof}

It is more natural, especially in terms of interpretation, to introduce the normalized spatially aggregated loss as a function of the cost field instead of its distribution, as shown immediately below.
\begin{Def}[Normalized spatially aggregated loss as a function of the cost field]
\label{Def_Norm_Spat_Agg_Loss_CP}
The normalized spatially aggregated loss function is defined by
\Beq
\begin{array}{ccccc}
L_N & : & \mathcal{A} \times \mathcal{C} & \to & \mathbb{R} \\
  &   & (A,C) & \mapsto & \dfrac{1}{\nu(A)} \displaystyle \int_{A} C(\Mb{x}) \  \nu(\mathrm{d}\Mb{x}).
\end{array}
\label{Eq_Def_Norm_Loss_Function_C}
\Eeq
\end{Def}
Let $C_P \in \mathcal{C}$ be a random field with distribution $P$. Although a particular realization of $L_N(A,C_P)$ obviously depends on $C_P$ (through its corresponding realization), we know from Theorem \ref{Th_Characterization_LN_Finite_Dimensional_Distributions_CP} that its distribution is entirely characterized by $A$ and $P$. This explains our notation $L_N(A,P)$ instead of $L_N(A,C_P)$ in Definition \ref{Def_normalized_Loss}. More precisely, let $C_P^{(1)}, C_P^{(2)} \in \mathcal{C}$ be random fields having the same distribution $P$. Then, $C_P^{(1)}$ and $C_P^{(2)}$ have the same finite-dimensional distributions, which implies that $L_N (A,C_P^{(1)}) \overset{d}{=} L_N (A, C_P^{(2)})$.  

Similarly, it can appear more natural to define spatial risk measures as functions of the cost field instead of its distribution. Moreover, this allows spatial risk measures to be defined even when the classical risk measure $\Pi$ is not law-invariant.
\begin{Def}[Spatial risk measure as a function of the cost field]
\label{Def_Spatial_Risk_Measure_Associated_With_Cp}
A spatial risk measure is a function $\mathcal{R}_{\Pi}$ that assigns a real number to any region $A \in \mathcal{A}$ and random field $C \in \mathcal{C}$:
$$\begin{array}{ccccc}
\mathcal{R}_{\Pi} & : & \mathcal{A} \times \mathcal{C} & \to & \mathbb{R} \\
  &   & (A,C) & \mapsto & \Pi ( L_N(A,C) ),
\end{array}$$
where $\Pi$ is a classical risk measure.
\end{Def}
For a given classical and law-invariant risk measure $\Pi$ and a given region $A \in \mathcal{A}$, the value of the spatial risk measure of Definition \ref{Def_Spatial_Risk_Measure_Associated_With_Cp} is completely determined by the distribution of $L_N(A, C)$ by law-invariance of $\Pi$. Consequently, using Theorem \ref{Th_Characterization_LN_Finite_Dimensional_Distributions_CP}, it is completely determined by $A$ and the distribution of the cost field $C$. This explains why \cite{koch2017spatial} has introduced the notion of spatial risk measure as a function of the distribution of $C$ (see the reminder in Definition \ref{Def_Spatial_Risk_Measure}); if $\Pi$ is law-invariant, the spatial risk measures described in Definitions \ref{Def_Spatial_Risk_Measure} and \ref{Def_Spatial_Risk_Measure_Associated_With_Cp} refer to the same notion. For a given $\Pi$ and a fixed $C \in \mathcal{C}$, $\mathcal{R}_{\Pi}(\cdot,C)$ is referred to as the spatial risk measure associated with $\Pi$ and induced by $C$. 

Of course, we can also express the axioms recalled in Definition \ref{Chapriskmeasures_Def_Axiomatic} for the spatial risk measures induced by a cost field $C \in \mathcal{C}$ introduced immediately above. On top of being more natural, it enables one to leave out the assumption of law-invariance for the classical risk measure $\Pi$.
\begin{Def}[Set of axioms for spatial risk measures induced by a cost field]
\label{Chapriskmeasures_Def_Axiomatic_Risk_Measure_Function_Cp}
Let $\Pi$ be a classical risk measure. For a fixed $C \in \mathcal{C}$, we define the following axioms for the spatial risk measure associated with $\Pi$ and induced by $C$, $\mathcal{R}_{\Pi}(\cdot, C)$:
\medskip
\Ben
\item \textbf{Spatial invariance under translation:} \\ for all $\Mb{v} \in \mathbb{R}^2$ and $A \in \mathcal{A},  \ \mathcal{R}_{\Pi}(A+\Mb{v}, C)=\mathcal{R}_{\Pi}(A,C)$, where $A+\Mb{v}$ denotes the region $A$ translated by the vector $\Mb{v}$.
\item \textbf{Spatial sub-additivity:} \\
for all $A_1, A_2 \in \mathcal{A},\  \mathcal{R}_{\Pi}(A_1 \cup A_2, C) \leq \min \{ \mathcal{R}_{\Pi}(A_1, C),\mathcal{R}_{\Pi}(A_2, C) \}$.
\item \textbf{Asymptotic spatial homogeneity of order $\boldsymbol{- \gamma}\boldsymbol{, \gamma \geq 0}$:} \\ for all $A \in \mathcal{A}_c$,
$$
\mathcal{R}_{\Pi}(\lambda A, C) \underset{\lambda \to \infty}{=} K_1(A,C)+\dfrac{K_2(A,C)}{\lambda^{\gamma}} + o\left(\frac{1}{\lambda^{\gamma}}\right),
$$
where $\lambda A$ is the area obtained by applying to $A$ a homothety with center $\Mb{b}_A$ and ratio $\lambda >0$, and $K_1(\cdot, C): \mathcal{A}_c  \to \mathbb{R}$, $K_2(\cdot, C): \mathcal{A}_c  \to \mathbb{R} \backslash \{ 0 \}$ are functions depending on $C$.
\Een
\end{Def}
All the explanations and interpretations given for Definitions \ref{Def_normalized_Loss}-\ref{Chapriskmeasures_Def_Axiomatic} remain valid in the case of Definitions \ref{Def_Norm_Spat_Agg_Loss_CP}-\ref{Chapriskmeasures_Def_Axiomatic_Risk_Measure_Function_Cp}. For the reasons mentioned above, our opinion is that Definitions \ref{Def_Norm_Spat_Agg_Loss_CP}-\ref{Chapriskmeasures_Def_Axiomatic_Risk_Measure_Function_Cp} rather than previous ones should be used. This is what is done in the following.

\subsection{Concrete applications to insurance}
\label{Subsec_Concrete_App_Ins_Reins_Industry}

This section is dedicated to the connections between the concepts described above and actuarial risk theory as well as real insurance practice. We especially show how they can be used for concrete purposes. In an insurance context, the quantity
\Beq
\label{Eq_Total_Loss}
L(A, C)=\int_{A} C(\Mb{x}) \  \nu(\mathrm{d}\Mb{x})
\Eeq
appearing in Definition \ref{Def_Norm_Spat_Agg_Loss_CP} (or equivalently \eqref{Eq_Loss_Function_Dist}) can be seen as a continuous and more complex version of the classical actuarial individual risk model. The latter can be formulated as 
\Beq
\label{Eq_Classical_Individual_Risk_Model}
L_{{\footnotesize \mbox{ind}}} = \sum_{i=1}^N X_i,
\Eeq
where $L_{{\footnotesize \mbox{ind}}}$ is the total loss, $N$ denotes the number of insurance policies and, for $i=1, \dots, N,$ $X_i$ is the claim related to the $i$-th policy. The $X_i$'s are generally assumed to be independent but not necessarily identically distributed. In $L(A,C)$, each location $\Mb{x}$ corresponds to a specific insurance policy and thus each $C(\Mb{x})$ is equivalent to a $X_i$ in \eqref{Eq_Classical_Individual_Risk_Model}. By the way, by choosing $\nu$ to be a counting measure instead of the Lebesgue measure, the integral in \eqref{Eq_Total_Loss} can be reduced to a sum, e.g., $\sum_{\Mb{x} \in A^{\prime}} C(\Mb{x})$, where $A^{\prime}$ is a finite set of locations in $\mathbb{R}^2$ (e.g., part of a lattice in $\mathbb{Z}^2$). It is worth mentioning that the ideas of this paper can easily be applied to such a framework. 

Even if dependence between the $X_i$, $i=1, \dots, N,$ in \eqref{Eq_Classical_Individual_Risk_Model} was allowed, considering $L(A,C)$ (see \eqref{Eq_Total_Loss}) would appear more promising. Indeed, the geographical information of each risk (i.e., insurance policy) is explicitly taken into account and, consequently, the dependence between all risks can be modelled in a more realistic way than in \eqref{Eq_Classical_Individual_Risk_Model}. The dependence between the risks directly inherits from their respective associated geographical positions and, thus, ignoring their localizations as in \eqref{Eq_Classical_Individual_Risk_Model} makes the modelling of their dependence more arbitrary and likely less reliable. In our approach, this dependence is fully characterized by the spatial dependence structure of the cost field $C$. Potential central limit theorems (see below) would have stronger implications because the dependence is more realistic. For these reasons, Models \eqref{Eq_Def_Norm_Loss_Function_C} and \eqref{Eq_Total_Loss} allow a more accurate assessment of spatial diversification. The same remarks hold if we compare our loss models with the classical actuarial collective risk model. 

Our risk models \eqref{Eq_Def_Norm_Loss_Function_C} and \eqref{Eq_Total_Loss} and more generally our theory about spatial risk measures may be particularly relevant for an insurance company willing to adapt its policies portfolio. E.g., the axioms of spatial sub-additivity and asymptotic spatial homogeneity can help it to assess the potential relevance of extending its activity to a new geographical region. Such an analysis requires the company to have an accurate view of the dependence between its risks (inter alia between the possible new risks and those already present in the portfolio), as allowed by Models \eqref{Eq_Def_Norm_Loss_Function_C} and \eqref{Eq_Total_Loss} through the cost field $C$. Model \eqref{Eq_Classical_Individual_Risk_Model} would not enable the insurer to precisely account for the dependence between the new risks and those already in the portfolio and hence to properly quantify the impact of a geographical expansion, i.e., of an increase of the number of contracts $N$.

At present, we show that, consistently with our intuition, considering the risk related to the normalized spatially aggregated loss is also insightful when the insurer is interested in the risk related to its non-normalized counterpart, which is often the case. Let $\Pi$ be a positive homogeneous and translation invariant classical risk measure and $p_r$ denote either the claims reserves, revenues or any relevant related quantity\footnote{It is out of the scope of this paper to enter into accounting details.} per surface unit (possibly the mean premium per surface unit) of an insurance company Ins. 

We first consider the axiom of spatial sub-additivity, which is assumed to be satisfied. Ins covers region $A_1$ for a given hazard and potentially aims at covering also a region $A_2$ disjoint of $A_1$. We assume that Ins properly hedges its risk on $A_1$, i.e., 
\Beq
\label{Eq_Condition_Claim_A1}
\nu(A_1) p_r \geq \Pi(L(A_1, C)), \quad \mbox{i.e., } \quad p_r \geq \Pi(L_N(A_1, C)),
\Eeq
by positive homogeneity. Using again the same property,
$$ \Pi(L(A_1 \cup A_2, C)) = \nu(A_1 \cup A_2) \Pi(L_N(A_1 \cup A_2, C)).$$
Combined with
$$ \Pi(L_N(A_1 \cup A_2, C)) \leq \Pi(L_N(A_1, C)),$$
this yields
$$ \Pi(L(A_1 \cup A_2, C)) \leq \frac{\nu(A_1 \cup A_2)}{\nu(A_1)} \Pi(L(A_1, C)).$$
Hence, by translation invariance,
\begin{align}
\Pi(L(A_1 \cup A_2, C)- \nu(A_1 \cup A_2)p_r)& =\Pi(L(A_1 \cup A_2, C))-\nu(A_1 \cup A_2)p_r \nonumber \\
& \leq \frac{\nu(A_1 \cup A_2)}{\nu(A_1)} \Pi(L(A_1, C)) -\nu(A_1 \cup A_2)p_r.
\label{Eq_Sub_Add_Light_1}
\end{align}
It follows from \eqref{Eq_Condition_Claim_A1} that
$$ p_r [\nu(A_1 \cup A_2) - \nu(A_1)] \geq \frac{\Pi(L(A_1, C))}{\nu(A_1)} [\nu(A_1 \cup A_2) - \nu(A_1)],$$
which gives
\Beq
\label{Eq_Sub_Add_Light_2}
\frac{\nu(A_1 \cup A_2)}{\nu(A_1)}\Pi(L(A_1, C))-\nu(A_1 \cup A_2)p_r \leq \Pi(L(A_1,C)) - \nu(A_1) p_r.
\Eeq
The combination of \eqref{Eq_Sub_Add_Light_1} and \eqref{Eq_Sub_Add_Light_2} yields that
$$ \Pi(L(A_1 \cup A_2, C)- \nu(A_1 \cup A_2)p_r) \leq \Pi(L(A_1,C) - \nu(A_1) p_r).$$
The last inequality is strict if that in the axiom of spatial sub-additivity or in \eqref{Eq_Condition_Claim_A1} is so.
Thus, if Ins suitably hedges its risk on $A_1$, the risk is even better hedged on $A_1 \cup A_2$. Exactly the same reasoning holds for $A_2$.
\begin{Rq}
\label{Rq_Subadd_Nonnormalized}
For spatial risk measures defined using the non-normalized spatially aggregated loss, we could propose the following axiom of spatial sub-additivity: for all disjoint $A_1, A_2 \in \mathcal{A}$, $\Pi(L(A_1 \cup A_2, C)) \leq \Pi(L(A_1, C)) + \Pi(L(A_2, C))$. Nevertheless, this property is trivially satisfied as soon as the classical risk measure $\Pi$ is sub-additive and therefore its validity does not depend on the properties of the cost field $C$. Basing the axiom of spatial sub-additivity on the normalized spatially aggregated loss as we did is more appealing since it allows a diversification effect coming from $C$ (and not only from $\Pi$). This argument is in favour of defining spatial risk measures using the normalized spatially aggregated loss.
\end{Rq}

We now consider the axiom of asymptotic spatial homogeneity of order $-\gamma$, $\gamma \geq 0$. Assume that it is satisfied with $\gamma>0$ (e.g., we will see that for $\Pi$ being VaR or ES, $\gamma$ typically equals $1$). It follows from Definition \ref{Chapriskmeasures_Def_Axiomatic_Risk_Measure_Function_Cp}, Point 3, that
\begin{align*}
\label{}
& \quad \quad \ \Pi( L(\lambda A, C)-\nu(\lambda A)p_r) 
\\&\underset{\lambda \to \infty}{=} \lambda^2 \nu(A) K_1(A, C) + \nu(A) K_2(A, C) \lambda^{2-\gamma} + o\left(\lambda^{2-\gamma} \right) - \lambda^2 \nu(A) p_r
\\& \underset{\lambda \to \infty}{=} \lambda^2 \nu(A) (K_1(A, C)-p_r) + \nu(A) K_2(A, C) \lambda^{2-\gamma} + o\left(\lambda^{2-\gamma} \right).
\end{align*}
Since $\gamma>0$, the dominant term as $\lambda \to \infty$ is $\lambda^2 \nu(A) (K_1(A, C)-p_r)$. Assume that $K_1(A,C)>0$ and $K_2(A,C)>0$. This is true under the conditions of Section \ref{Sec_Axioms} for VaR and ES: we have $K_1(A, C)=\mathbb{E}[C(\Mb{0})]$, which is positive as the cost field can be assumed to be non-negative and not a.s. equal to $0$; regarding $K_2(A,C)$, this is always true for ES and, provided that the confidence level $\alpha$ is greater than $1/2$, also for VaR. Consequently, for $\lambda$ large enough, the total risk of the company, $\Pi( L(\lambda A, C)-\nu(\lambda A)p_r)$, is a decreasing function of $\lambda$ as soon as the revenue per surface unit (or claims reserves, \dots) satisfies $p_r > K_1(A, C)$. Under the conditions of Section \ref{Sec_Axioms}, for VaR and ES, $K_1(A, C)=\mathbb{E}[C(\Mb{x})]$ for all $\Mb{x} \in \mathbb{R}^2$, and therefore the latter inequality entails that the revenue per surface unit (e.g., the mean premium) exceeds the expected cost at each location, which appears natural. The term $2-\gamma$ corresponds to the second highest power with respect to $\lambda$. Provided that $K_2(A,C)>0$ and $0 < \gamma < 2$ (which is true for VaR and ES under the conditions of Section \ref{Sec_Axioms}), the corresponding term, $\nu(A) K_2(A, C) \lambda^{2-\gamma}$, increases the total risk of the company as $\lambda$ increases. However, the highest the value of $\gamma$, the fastest the decrease of the total risk as $\lambda$ increases owing to the term in $\lambda^2$. For $\lambda$ large, the values of $\gamma$, $K_1(A,C), K_2(A,C)$ and $p_r$ allow one to determine the value of $\lambda$ necessary to reach a targeted sufficiently low level of the total risk. Note that in the case of the variance, at least under the conditions of Section \ref{Sec_Axioms}, $K_1(A,C)=0$ and $\gamma=2$.

\begin{Rq}
The axioms of spatial invariance under translation and asymptotic spatial homogeneity could also be defined for spatial risk measures based on the non-normalized spatially aggregated loss. Spatial invariance under translation would be unchanged and asymptotic spatial homogeneity of order $- \gamma$, $\gamma \geq 0$, would become: for all $A \in \mathcal{A}_c$,
$$
\Pi( L(\lambda A, C) ) \underset{\lambda \to \infty}{=} \lambda^2 \nu(A) K_1(A, C) + \nu(A) K_2(A, C) \lambda^{2-\gamma} + o\left(\lambda^{2-\gamma} \right).
$$
In this case, we would obtain the risk related to the non-normalized loss without assuming that $\Pi$ is positive homogeneous.
\end{Rq}

Finally, we discuss a possible way for a company to develop an adequate model for the cost field $C$ in regions where it is still inactive. The general model for the cost field introduced in \cite{koch2017spatial}, Section 2.3, is written
\Beq
\label{Eq_Economic_Loss_Model}
\left \{ C(\Mb{x}) \right \}_{\Mb{x} \in \mathbb{R}^2} = \left \{ E(\Mb{x})\  D \left( Z(\Mb{x}) \right) \right \}_{\Mb{x} \in \mathbb{R}^2},
\Eeq
where $\{ E(\Mb{x}) \}_{\Mb{x} \in \mathbb{R}^2}$ is the exposure field, $D$ a damage function and $\{ Z(\Mb{x}) \}_{\Mb{x} \in \mathbb{R}^2}$ the random field of the environmental variable generating risk. The cost is assumed to be only due to a unique class of events, i.e., to a unique natural hazard. The latter (e.g., heat waves or hurricanes) is described by the random field of an environmental variable (e.g., the temperature or the wind speed, respectively), $Z$. We assume that $Z$ is representative of the risk during the whole period $[0, T_L]$. The application of the damage function (also referred to as vulnerability curve in the literature) $D$ to the natural hazard random field gives the destruction percentage at each location. Finally, multiplying the destruction percentage by the exposure gives the cost at each location. For more details, we refer the reader to \cite{koch2017spatial}, Section 2.3.
In order to obtain an adequate model $C$ in regions where it has no policies yet, the company can for instance consider crude estimates of the exposure field in the new region, develop a detailed statistical model\footnote{Potentially different from those developed in the natural catastrophes industry: e.g., a max-stable model.} for the environmental field $Z$ responsible of the risk insured (e.g., wind speed in the case of hurricanes) using appropriate data and apply the same damage functions as in the region it already covers. The company can then simulate from this cost model, hence obtaining an empirical distribution of the loss appearing in \eqref{Eq_Def_Norm_Loss_Function_C} and \eqref{Eq_Total_Loss}. This makes it possible to check whether the axiom of spatial sub-additivity is satisfied or not. Furthermore, if the spatial domain is large (which is generally the case for reinsurance companies), considering potential central limit theorems and determining the order of asymptotic spatial homogeneity (by checking if the conditions of Section \ref{Sec_Axioms} are satisfied) is useful as it allows the company to quantify the rate of spatial diversification.
\begin{Rq}
Strictly speaking, the terms of the insurance policies should be accounted for in Model \eqref{Eq_Economic_Loss_Model}. By the way, the latter model can be interpreted differently from what is done here. For instance, we can imagine that $Z$ represents the random field of the real cost and $D$ accounts for the terms of the policies.
\end{Rq}

\subsection{Mixing and central limit theorems for random fields}

We first remind the reader of the definition of the $\alpha$- and $\beta$-mixing coefficients which will be used in Section \ref{Sec_Axioms}.
Let $\{ X(\Mb{x}) \}_{\Mb{x} \in \mathbb{R}^d}$ be a real-valued random field. For $S \subset \mathbb{R}^d$ a closed subset, we denote by $\mathcal{F}^X_S$ the $\sigma$-field generated by the random variables $\{ X(\Mb{x}): \Mb{x} \in S \}$. Let $S_1, S_2 \subset \mathbb{R}^d$ be disjoint closed subsets.  The $\alpha$-mixing coefficient \citep[introduced by][]{rosenblatt1956central} between the $\sigma$-fields $\mathcal{F}^X_{S_1}$ and $\mathcal{F}^X_{S_2}$ is defined by
\Beq
\label{Eq_Def_Alpha_Mixing_Coefficient}
\alpha^X(S_1,S_2)=\sup\Big\{|\mathbb{P}(A\cap B)-\mathbb{P}(A)\mathbb{P}(B)|: A\in \mathcal{F}^X_{S_1}, B\in \mathcal{F}^X_{S_2} \Big\}.
\Eeq
The $\beta$-mixing coefficient or absolute regularity coefficient \citep[attributed to Kolmogorov in][]{volkonskii1959some} between the $\sigma$-fields $\mathcal{F}_{S_1}^X$ and $\mathcal{F}_{S_2}^X$ is given by
$$
\beta^X(S_1, S_2)= \frac{1}{2} \sup \left \{ \sum_{i=1}^{I} \sum_{j=1}^J | \mathbb{P}(A_i \cap B_j) - \mathbb{P}(A_i) \mathbb{P}(B_j) | \right \},
$$
where the supremum is taken over all partitions $\{ A_1,\dots , A_I \}$ and $\{ B_1, \dots , B_J \}$ of $\Omega$ with the
$A_i$'s in $\mathcal{F}_{S_1}^X$ and the $B_j$'s in $\mathcal{F}_{S_2}^X$. These coefficients satisfy the useful inequality 
\Beq
\label{Eq_Majoration_Alphamixing_With_Betamixing}
\alpha^X(S_1, S_2) \leq \frac{1}{2} \beta^X(S_1, S_2), \quad \mbox{for all } S_1, S_2 \subset \mathbb{R}^d.
\Eeq

Now, we recall the concepts of Van Hove sequence and central limit theorem (CLT) in the case of random fields. This will be useful, since, for instance, asymptotic spatial homogeneity of order $-1$ of spatial risk measures associated with VaR (at a confidence level $\alpha \in (0,1) \backslash \{ 1/2 \}$) and induced by a cost field $C \in \mathcal{C}$ is satisfied as soon as $C$ fulfills the CLT and has a constant expectation (see below).
For $V \subset \mathbb{R}^d$ and $r>0$, we introduce $V^{+r}=\{ \Mb{x} \in \mathbb{R}^d: \mathrm{dist}(\Mb{x}, V) \leq r \}$, where $\mathrm{dist}$ stands for the Euclidean distance. Additionally, we denote by $\partial V$ the boundary of $V$. A Van Hove sequence in $\mathbb{R}^d$ is a sequence $( V_n )_{n \in \mathbb{N}}$ of bounded measurable subsets of $\mathbb{R}^d$ satisfying $V_n \uparrow \mathbb{R}^d$, $\lim_{n \to \infty} \nu(V_n)=\infty$, and $\lim_{n \to \infty} \nu((\partial V_n)^{+r} )/\nu(V_n) =0 \mbox{ for all } r>0$. The assumption ``bounded'' does not always appear in the definition of a Van Hove sequence. Let Cov denote the covariance. In the following, we say that a random field $\{ X(\Mb{x}) \}_{\Mb{x} \in \mathbb{R}^d}$ such that, for all $\Mb{x} \in \mathbb{R}^d$, $\mathbb{E}\left[ [X(\Mb{x})]^2 \right]< \infty$, satisfies the CLT if
$$
\int_{\mathbb{R}^d} | \mathrm{Cov}(X(\Mb{0}), X(\Mb{x})) | \  \nu(\mathrm{d}\Mb{x}) < \infty,
$$ 
$$
\sigma_X= \left( \int_{\mathbb{R}^d} \mathrm{Cov}(X(\Mb{0}), X(\Mb{x})) \  \nu(\mathrm{d}\Mb{x}) \right)^{\frac{1}{2}}>0,
$$
and, for any Van Hove sequence $(V_n)_{n \in \mathbb{N}}$ in $\mathbb{R}^d$,
$$
\frac{1}{[\nu(V_n)]^{1/2}} \int_{V_n} (X(\Mb{x})-\mathbb{E}[X(\Mb{x})]) \  \nu(\mathrm{d}\Mb{x}) \overset{d}{\rightarrow} \mathcal{N}(0, \sigma_X^2), \quad \mbox{as } n\to\infty,
$$
where $\mathcal{N}(\mu, \sigma^2)$ denotes the normal distribution with expectation $\mu \in \mathbb{R}$ and variance $\sigma^2>0$. In the case of a random field satisfying the CLT, we have the following result.
\begin{Th}
\label{Th_Link_CLT_VanHove_CLT_Homothety}
Let $\{ C(\Mb{x}) \}_{\Mb{x} \in \mathbb{R}^2} \in \mathcal{C}$. Assume moreover that $C$ has a constant expectation (i.e., for all $\Mb{x} \in \mathbb{R}^2$, $\mathbb{E}[C(\Mb{x})]=\mathbb{E}[C(\Mb{0})]$) and satisfies the CLT. Then, we have, for all $A \in \mathcal{A}_c$, that
$$
\lambda \left( L_N(\lambda A, C) - \mathbb{E}[C(\Mb{0})] \right) \overset{d}{\to} \mathcal{N} \left( 0, \frac{\sigma_{C}^2}{\nu(A)} \right), \mbox{ for } \lambda \to \infty.
$$
\end{Th}

\begin{proof}
The result is essentially based on part of the proof of Theorem 4 in \cite{koch2017spatial}. We refer the reader to this proof for details and only provide the main ideas here. First, we show \citep[see][third paragraph of the proof of Theorem 4]{koch2017spatial} that, for any $A \in \mathcal{A}_c$ and any positive non-decreasing sequence $(\lambda_n)_{n \in \mathbb{N}} \in \mathbb{R}$ such that $\lim_{n \to \infty} \lambda_n=\infty$, the sequence $(\lambda_n A)_{n \in \mathbb{N}}$ is a Van Hove sequence. Therefore, since $C$ satisfies the CLT and has a constant expectation, we obtain
$$ \lambda_n \left( L_N(\lambda_n A, C) - \mathbb{E}[C(\Mb{0})] \right) \overset{d}{\to} \mathcal{N} \left( 0, \frac{\sigma_{C}^2}{\nu(A)} \right), \mbox{ for } n \to \infty.$$
Second, we deduce \citep[see][proof of Theorem 4, after (44)]{koch2017spatial} that, for all $A \in \mathcal{A}_c$,
$$
\lambda \left( L_N(\lambda A, C) - \mathbb{E}[C(\Mb{0})] \right) \overset{d}{\to} \mathcal{N} \left( 0, \frac{\sigma_{C}^2}{\nu(A)} \right), \mbox{ for } \lambda \to \infty.
$$
This concludes the proof.
\end{proof}
This theorem will be useful in the following since it will allow us to prove asymptotic spatial homogeneity of order respectively $-2$, $-1$ and $-1$ for spatial risk measures associated with variance, VaR as well as ES and induced by a cost field satisfying the CLT and additional conditions. Moreover, if $\lambda$ is large enough, it gives an approximation of the distribution of the normalized spatially aggregated loss:
$$ L_N(\lambda A, C) \approx \mathcal{N} \left( \mathbb{E}[C(\Mb{0})],  \frac{\sigma_{C}^2}{\lambda^2 \nu(A)} \right),$$
where $\approx$ means ``approximately follows''. Such an approximation can be fruitful in practice, e.g., for an insurance company. 

\subsection{Max-stable random fields}

This concise introduction to max-stable fields is partly based on \cite{koch2017TCL}, Section 2.2. Below, ``$\bigvee$'' denotes the supremum when the latter is taken over a countable set. In any dimension $d \geq 1$, max-stable random fields are defined as follows.
\begin{Def}[Max-stable random field]
\label{Def_Maxstable_Processes}
A real-valued random field $\left \{ Z(\Mb{x}) \right \}_{\Mb{x} \in \mathbb{R}^d}$ is said to be max-stable if there exist sequences of functions 
$( a_T(\Mb{x}), \Mb{x} \in \mathbb{R}^d)_{T \geq 1}> 0$ and
$(  b_T(\Mb{x}), \Mb{x} \in \mathbb{R}^d)_{T \geq 1} \in \mathbb{R}$
such that, for all $T \geq 1$,
$$
\left \{ \frac{ \bigvee_{t=1}^T \left \{ Z_t(\Mb{x}) \right \}-b_T(\Mb{x} )}{a_T(\Mb{x} )}  \right \} _{\Mb{x} \in \mathbb{R}^d}  \overset{d}{=} \{ Z(\Mb{x}) \}_{\Mb{x} \in \mathbb{R}^d},
$$
where the $\{ Z_t(\Mb{x})\}_{\Mb{x} \in \mathbb{R}^d }, t=1, \dots, T,$ are independent replications of $Z$.
\end{Def}
A max-stable random field is termed simple if it has standard Fr\'echet margins, i.e., for all $\Mb{x} \in \mathbb{R}^d$, $\mathbb{P}(Z(\Mb{x}) < z)=\exp \left( -1/z \right), z>0$.

Now, let $\{ \tilde{T_i}(\Mb{x}) \}_{\Mb{x} \in \mathbb{R}^d}, i=1, \dots, n,$ be independent replications of a random field $ \{ \tilde{T}(\Mb{x}) \}_{\Mb{x} \in \mathbb{R}^d}$. Let 
$( c_n(\Mb{x}), \Mb{x} \in \mathbb{R}^d)_{n \geq 1} >0$ and $( d_n(\Mb{x}), \Mb{x} \in \mathbb{R}^d)_{n \geq 1} \in \mathbb{R}$ be sequences of functions. If there exists a non-degenerate random field $\{ G(\Mb{x}) \}_{\Mb{x} \in \mathbb{R}^d}$ such that
$$
\left \{ \frac{\bigvee_{i=1}^n \left \{ \tilde{T_i}(\Mb{x}) \right \} -d_n(\Mb{x})}{c_n(\Mb{x})} \right \}_{\Mb{x} \in \mathbb{R}^d} \overset{d}{\rightarrow}  \left \{ G(\Mb{x}) \right \}_{\Mb{x} \in \mathbb{R}^d}, \mbox{ for } n \to \infty,
$$
then $G$ is necessarily max-stable; see, e.g., \cite{haan1984spectral}. This explains the relevance and significance of max-stable random fields in the modelling of spatial extremes. 

Any simple max-stable random field $Z$ can be written \citep[see, e.g.,][]{haan1984spectral} as
\Beq
\label{Eq_Spectral_Representation_Stochastic_Processes}
\left \{ Z(\Mb{x}) \right \}_{\Mb{x} \in \mathbb{R}^d} \overset{d}{=} \left \{ \bigvee_{i=1}^{\infty} \{ U_i Y_i(\Mb{x}) \} \right \}_{\Mb{x} \in \mathbb{R}^d},
\Eeq
where the $(U_i)_{i \geq 1}$ are the points of a Poisson point process on $(0, \infty)$ with intensity $u^{-2} \nu(\mathrm{d}u)$ and the $Y_i, i\geq 1$, are independent replications of a random field $\{ Y(\Mb{x}) \}_{\Mb{x} \in \mathbb{R}^d}$ such that, for all $\Mb{x} \in \mathbb{R}^d$,
$\mathbb{E}[Y(\Mb{x})]=1$. The field $Y$ is not unique and is called a spectral random field of $Z$. Conversely, any random field of the form \eqref{Eq_Spectral_Representation_Stochastic_Processes} is a simple max-stable random field. Hence, \eqref{Eq_Spectral_Representation_Stochastic_Processes} enables the building up of models for max-stable fields. We now present one of the most famous among such models, the Brown--Resnick random field, which is defined in \cite{kabluchko2009stationary} as a generalization of the stochastic process introduced in \cite{brown1977extreme}. We recall that a random field $\{ W(\Mb{x}) \}_{\Mb{x} \in \mathbb{R}^d}$ is said to have stationary increments if the distribution of the random field $\{ W(\Mb{x}+\Mb{x}_0)-W(\Mb{x}_0) \}_{\Mb{x} \in \mathbb{R}^d}$ does not depend on $\Mb{x}_0 \in \mathbb{R}^d$. Provided the increments of $W$ have a finite second moment, the variogram of $W$, $\gamma_W$, is defined by 
$$ \gamma_W(\Mb{x})=\mathrm{Var}(W(\Mb{x})-W(\Mb{0})), \quad \Mb{x} \in \mathbb{R}^d,$$
where Var denotes the variance. The Brown--Resnick random field is specified as follows.
\begin{Def}[Brown--Resnick random field]
\label{Def_Brown_Resnick_Process}
Let $\{ W(\Mb{x}) \}_{\Mb{x} \in \mathbb{R}^d}$ be a centred Gaussian random field with stationary increments and with variogram $\gamma_W$. Let us consider the random field $Y$ defined by
$$ \{ Y(\Mb{x}) \}_{\Mb{x} \in \mathbb{R}^d}=\left \{ \exp \left( W(\Mb{x})-\frac{\mathrm{Var}(W(\Mb{x}))}{2} \right) \right \}_{\Mb{x} \in \mathbb{R}^d}.$$
Then the simple max-stable random field defined by \eqref{Eq_Spectral_Representation_Stochastic_Processes} with $Y$ is referred to as the Brown--Resnick random field associated with the variogram $\gamma_W$. In the following, we will call this field the Brown--Resnick random field built with $W$.\footnote{In the following, when $W$ is sample-continuous, what we refer to as the Brown--Resnick random field built with $W$ is obtained by taking replications of $W$ (see \eqref{Eq_Spectral_Representation_Stochastic_Processes}) which are also sample-continuous.}
\end{Def}
The Brown--Resnick field is stationary \citep[][Theorem 2]{kabluchko2009stationary} and its distribution only depends on the variogram \citep[][Proposition 11]{kabluchko2009stationary}. 

Now, let $(U_i, \Mb{C}_i)_{i \geq 1}$ be the points of a Poisson point process on $(0,\infty) \times \mathbb{R}^d$ with intensity function $u^{-2} \nu(\mathrm{d}u) \times \nu(\mathrm{d}\Mb{c})$. Independently, let $f_i, i \geq 1$, be independent replicates of some non-negative random function $f$ on $\mathbb{R}^d$ satisfying $\mathbb{E} \left[ \int_{\mathbb{R}^d} f(\Mb{x}) \ \nu(\mathrm{d}\Mb{x}) \right]=1$. Then, it is known that the Mixed Moving Maxima (M3) random field
\Beq
\label{Eq_Mixed_Moving_Maxima_Representation}
\{ Z(\Mb{x}) \}_{\Mb{x} \in \mathbb{R}^d}= \left \{ \bigvee_{i=1}^{\infty} \{ U_i f_i(\Mb{x}-\Mb{C}_i) \} \right \}_{\Mb{x} \in \mathbb{R}^d}
\Eeq
is a stationary and simple max-stable field. The so called Smith random field introduced by \cite{smith1990max} is a specific case of M3 random field and is defined immediately below.
\begin{Def}[Smith random field]
Let $Z$ be written as in \eqref{Eq_Mixed_Moving_Maxima_Representation} with $f$ being the density of a $d$-variate Gaussian random vector with mean $\Mb{0}$ and covariance matrix $\Sigma$. Then, the field $Z$ is referred to as the Smith random field with covariance matrix $\Sigma$.
\end{Def}
As the Brown--Resnick and Smith fields are defined using the random fields-based and M3 representations \eqref{Eq_Spectral_Representation_Stochastic_Processes} and \eqref{Eq_Mixed_Moving_Maxima_Representation}, respectively, it is usual in the spatial extremes literature to distinguish both models, although the Smith field with covariance matrix $\Sigma$ corresponds to the Brown-Resnick field associated with the variogram $\gamma_W(\Mb{x})=\Mb{x}^{'} \Sigma^{-1} \Mb{x}, \Mb{x} \in \mathbb{R}^d$, where ${}^{\prime}$ designates transposition; see, e.g., \cite{huser2013composite}.

Finally, we briefly present the extremal coefficient \citep[see, e.g.,][]{schlather2003dependence} which is a well-known measure of spatial dependence for max-stable random fields. Let $\{ Z(\Mb{x}) \}_{\Mb{x} \in \mathbb{R}^d}$ be a simple max-stable random field. In the case of two locations, the extremal coefficient function $\theta$ is defined by
$$
\mathbb{P}\left( Z(\Mb{x}_1) \leq u, Z(\Mb{x}_2) \leq u \right) = \exp \left( -\frac{\theta(\Mb{x}_1, \Mb{x}_2)}{u} \right), \quad \Mb{x}_1, \Mb{x}_2 \in \mathbb{R}^d,
$$
where $u>0$.

\section{Properties of some induced spatial risk measures}
\label{Sec_Axioms}

In this section, we provide sufficient conditions on the cost field such that some induced spatial risk measures satisfy the axioms presented in Definition \ref{Chapriskmeasures_Def_Axiomatic_Risk_Measure_Function_Cp}.  
First, we consider the case of a general cost field before investigating the relevant case of a cost field being a function of a max-stable random field. In the following, for $\alpha \in (0,1)$, $q_{\alpha}$ and $\phi$ denote the quantile at level $\alpha$ and the standard Gaussian density, respectively. We recall that for a random variable $\tilde{X}$ with distribution function $F$, its Value-at-Risk at confidence level $\alpha \in (0,1)$ is written $\mbox{VaR}_{\alpha} (\tilde{X} )=\inf \{ x \in \mathbb{R}: F(x) \geq \alpha \}$. Moreover, provided $\mathbb{E} [ | \tilde{X} |]<\infty$, its expected shortfall at confidence level $\alpha \in(0,1)$ is defined as 
$$ \mathrm{ES}_{\alpha}\left(\tilde{X}\right)= \frac{1}{1-\alpha} \int_{\alpha}^{1} \mathrm{VaR}_{u}\left(\tilde{X}\right) \  \nu(\mathrm{d}u).$$
Typical values for $\alpha$ are $0.95$ and $0.99$. It should be noted that in the actuarial literature, ES is sometimes referred to as Tail Value-at-Risk \citep[see, e.g.,][Definition 2.4.1]{denuit2005actuarial}.
In the following, we mainly consider the spatial risk measures
\begin{align*}
& \mathcal{R}_1(A,C) =\mathbb{E}[L_N(A,C)], \quad A \in \mathcal{A}, C \in \mathcal{C},
\\ & \mathcal{R}_2(A,C) =\mathrm{Var}(L_N(A,C)), \quad A \in \mathcal{A}, C \in \mathcal{C},
\\ & \mathcal{R}_{3, \alpha}(A,C) =\mathrm{VaR}_{\alpha}(L_N(A,C)), \quad A \in \mathcal{A}, C \in \mathcal{C},
\\& \mathcal{R}_{4, \alpha}(A,C)=\mathrm{ES}_{\alpha}(L_N(A,C)), \quad A \in \mathcal{A}, C \in \mathcal{C}.
\end{align*}
As a classical risk measure, the expectation is not very satisfying since it does not provide any information about variability. Moreover, as will be seen, the associated spatial risk measures do not take into account the spatial dependence of the cost field. An advantage of variance, VaR and ES lies in the fact that their associated spatial risk measures all take into account (at least) part of this spatial dependence. Historically, the variance has been the dominating risk measure in finance, primarily due to the huge influence of the portfolio theory of Markowitz which uses variance as a measure of risk. However, using the variance is only possible when the normalized spatially aggregated loss has a finite second moment. Moreover, since it allocates the same weight to positive and negative deviations from the expectation, variance is a good risk measure only for distributions which are approximately symmetric around the expectation. Currently, VaR is probably the most widely used risk measure in the finance/insurance industry. However, it does not provide any information about the severity of losses which occur with a probability lower than $1-\alpha$, which is obviously a serious shortcoming. Moreover, VaR is in general not sub-additive and hence not coherent in the sense of \cite{artzner1999coherent}. ES overcomes these two drawbacks of VaR. Pertaining to the first one, it can be seen from the fact that, if a random variable $\tilde{X}$ has a continuous distribution function, then
$$ \mathrm{ES}_{\alpha}\left(\tilde{X}\right)=\mathbb{E} \left [ \tilde{X} \Big | \tilde{X} > \mathrm{VaR}_{\alpha}(\tilde{X}) \right ].$$
Hence, the Basel Committee on Banking Supervision proposed the use of ES instead of VaR for the internal models-based approach \citep[][Section 3.2.1]{BaselCommittee2012}. However, contrary to VaR, ES is not elicitable \citep{gneiting2011making}, implying that backtesting for ES is more difficult than for VaR.

\subsection{General cost field}

Next result provides sufficient conditions on the cost field $C$ such that the induced spatial risk measure $\mathcal{R}_{1}(\cdot, C)$ satisfies the axioms presented in Definition \ref{Chapriskmeasures_Def_Axiomatic_Risk_Measure_Function_Cp}.
\begin{Th}
\label{Th_R1_General_Cost_Field}
Let $\{ C(\Mb{x}) \}_{\Mb{x} \in \mathbb{R}^2}$ be a measurable random field having a constant expectation and such that, for all $\Mb{x} \in \mathbb{R}^2$, $\mathbb{E}[|C(\Mb{x})|]=\mathbb{E}[|C(\Mb{0})|] < \infty$. Then, we have, for all $A \in \mathcal{A}$, that $\mathcal{R}_1(A, C)=\mathbb{E}[C(\Mb{0})]$. Hence, the spatial risk measure induced by $C$ $\mathcal{R}_1(\cdot, C)$ satisfies the axioms of spatial invariance under translation and spatial sub-additivity. If, moreover, $\mathbb{E}[C(\Mb{0})] \neq 0$, then $\mathcal{R}_1(\cdot, C)$ satisfies the axiom of asymptotic spatial homogeneity of order $0$ with $K_1(A, C)=0$ and $K_2(A, C)=\mathbb{E}[C(\Mb{0})], A \in \mathcal{A}_c$.
\end{Th}

\begin{proof}
By assumption, the function $\Mb{x} \mapsto \mathbb{E}[|C(\Mb{x})|]$ is constant and hence obviously locally integrable. Consequently, as $C$ is measurable, Proposition \ref{Eq_Sufficient_Condition_Locally_Integrable_Sample_Paths} gives that $C$ has a.s. locally integrable sample paths. Using Fubini's theorem and the fact that $C$ has a constant expectation, we have, for all $A \in \mathcal{A}$, that
$$\mathcal{R}_1(A,C)= \dfrac{1}{\nu(A)} \displaystyle\int_{A} \mathbb{E}[C(\Mb{0})] \  \nu(\mathrm{d}\Mb{x})=\mathbb{E}[C(\Mb{0})].$$
Thus, for all $\Mb{v} \in \mathbb{R}^2$ and $A \in \mathcal{A}$, $\mathcal{R}_1(A+\Mb{v},C)=\mathcal{R}_1(A,C)$.
Moreover, for all $A_1, A_2 \in \mathcal{A}$, 
$$\mathcal{R}_1(A_1 \cup A_2 ,C)=\mathcal{R}_1(A_1,C)=\mathcal{R}_1(A_2,C)=\min \{ \mathcal{R}_1(A_1,C), \mathcal{R}_1(A_2,C) \}.$$
Finally, for all $A \in \mathcal{A}_c$ and $\lambda>0$, we have
$\mathcal{R}_1(\lambda A,C)=\mathbb{E}[C(\Mb{0})]$. As $| \mathbb{E}[C(\Mb{0})] | \leq \mathbb{E}[|C(\Mb{0})|] < \infty$, we have $|\mathbb{E}[C(\Mb{0})]| \in (0, \infty)$, which concludes the proof.
\end{proof}

Next result is a generalization of Theorem 2 in \cite{koch2017spatial} and will be useful in the following.
\begin{Th}
\label{Th_R2_General}
Let $\{ C(\Mb{x}) \}_{\Mb{x} \in \mathbb{R}^2} \in \mathcal{C}$ and such that, for all $\Mb{x} \in \mathbb{R}^2$, $\mathbb{E} \left[ [ C(\Mb{x}) ]^2 \right]< \infty$. 
Moreover, assume that, for all $A \in \mathcal{A}$,
\Beq
\label{Eq_Condition_Existence_Expectation_Product}
\int_A \int_A | \mathbb{E} \left[ C(\Mb{x}) C(\Mb{y}) \right] | \  \nu(\mathrm{d}\Mb{x}) \  \nu(\mathrm{d}\Mb{y}) < \infty.
\Eeq
Then, for all $A \in \mathcal{A}$ and $\lambda>0$, we have
$$
\mathcal{R}_2(\lambda A, C)= \frac{1}{\lambda^4 [\nu(A)]^2} \int_{\lambda A}  \int_{\lambda A} \mathrm{Cov}(C(\Mb{x}), C(\Mb{y})) \  \nu(\mathrm{d}\Mb{x})\  \nu(\mathrm{d}\Mb{y}).
$$
Condition \eqref{Eq_Condition_Existence_Expectation_Product} is satisfied for instance in the following cases:
\Ben
\item For any $A \in \mathcal{A}$, 
\Beq
\label{Eq_Condition_Sup_Second_Moment}
\sup_{\Mb{x} \in A} \left \{ \mathbb{E} \left[ [ C(\Mb{x}) ]^2 \right] \right \} < \infty.
\Eeq
\item 
For all $\Mb{x}, \Mb{y} \in \mathbb{R}^2$, 
\Beq
\label{Eq_Dependence_Covariance_x-y}
\mathrm{Cov}(C(\Mb{x}), C(\Mb{y}))=\mathrm{Cov}(C(\Mb{0}), C(\Mb{x}-\Mb{y})),
\Eeq
and
\Beq
\label{Eq_Convergence_Abs_Int_Cov_R2}
\int_{\mathbb{R}^2} | \mathrm{Cov}(C(\Mb{0}), C(\Mb{x})) | \  \nu(\mathrm{d}\Mb{x}) < \infty.
\Eeq 
\Een
\end{Th}

\begin{proof}
For all $A \in \mathcal{A}$, we consider $L(A, C)=\nu(A) \ L_N(A, C)$.
Thus, using Fubini's theorem and \eqref{Eq_Condition_Existence_Expectation_Product}, we obtain
\begin{align}
\mathbb{E} \left[ [ L(A, C) ]^2 \right] = \mathbb{E} \left[ \left( \int_A C(\Mb{x}) \ \nu(\mathrm{d}\Mb{x})  \right)^2 \right] &= \mathbb{E} \left[ \int_A C(\Mb{x}) \ \nu(\mathrm{d}\Mb{x}) \ \int_A C(\Mb{y}) \ \nu(\mathrm{d}\Mb{y})  \right] \nonumber
\\&= \int_A \int_A \mathbb{E} \left[ C(\Mb{x}) C(\Mb{y}) \right] \  \nu(\mathrm{d}\Mb{x}) \  \nu(\mathrm{d}\Mb{y}).
\label{Eq_Expectation_LossSquare}
\end{align}
Moreover, it is clear that 
\Beq
\label{Eq_Expectation_Loss_Square}
\left( \mathbb{E} \left[ L(A, C) \right] \right)^2 = \int_A \int_A \mathbb{E} \left[ C(\Mb{x}) \right] \mathbb{E} \left[ C(\Mb{y}) \right] \  \nu(\mathrm{d}\Mb{x}) \ \nu(\mathrm{d}\Mb{y}).
\Eeq
The combination of \eqref{Eq_Expectation_LossSquare} and \eqref{Eq_Expectation_Loss_Square} gives that
$$ \mathbb{E} \left[ [ L(A, C) ]^2 \right]- \left( \mathbb{E} \left[ L(A, C) \right] \right)^2 = \int_{A}  \int_{A} \mathrm{Cov}(C(\Mb{x}), C(\Mb{y})) \  \nu(\mathrm{d}\Mb{x})\  \nu(\mathrm{d}\Mb{y}),$$
which implies that
$$ \mathcal{R}_2(A, C)= \frac{1}{[\nu(A)]^2} \int_{A}  \int_{A} \mathrm{Cov}(C(\Mb{x}), C(\Mb{y})) \  \nu(\mathrm{d}\Mb{x})\  \nu(\mathrm{d}\Mb{y}).$$
The result is obtained by replacing $A$ with $\lambda A$.

\medskip

We now prove the second part of the theorem, concerning \eqref{Eq_Condition_Existence_Expectation_Product}. Let $A \in \mathcal{A}$. In the first case, we obtain, using Cauchy--Schwarz inequality and \eqref{Eq_Condition_Sup_Second_Moment},
\begin{align*}
& \quad \ \int_A \int_A | \mathbb{E} \left[ C(\Mb{x}) C(\Mb{y}) \right] | \  \nu(\mathrm{d}\Mb{x}) \  \nu(\mathrm{d}\Mb{y}) 
\\&\leq \int_A \int_A \left( \mathbb{E} \left[ [ C(\Mb{x}) ]^2 \right ] \right)^{\frac{1}{2}} \left( \mathbb{E} \left[ [ C(\Mb{y}) ]^2 \right ] \right)^{\frac{1}{2}} \  \nu(\mathrm{d}\Mb{x}) \  \nu(\mathrm{d}\Mb{y})
\\& \leq \int_A \int_A \left( \sup_{\Mb{x} \in A} \left \{ \mathbb{E} \left[ [ C(\Mb{x}) ]^2 \right] \right \}
\right)^{\frac{1}{2}} \left( \sup_{\Mb{x} \in A} \left \{ \mathbb{E} \left[ [ C(\Mb{y}) ]^2 \right] \right \}
\right)^{\frac{1}{2}} \  \nu(\mathrm{d}\Mb{x}) \  \nu(\mathrm{d}\Mb{y})
\\& < \infty.
\end{align*}
 
In the second case, it follows from \eqref{Eq_Dependence_Covariance_x-y} and \eqref{Eq_Convergence_Abs_Int_Cov_R2} that
\begin{align*}
\int_{A}  \int_{A} | \mathrm{Cov}(C(\Mb{x}), C(\Mb{y})) | \  \nu(\mathrm{d}\Mb{x})\  \nu(\mathrm{d}\Mb{y})
&= \int_{A}  \int_{A} |\mathrm{Cov}(C(\Mb{0}), C(\Mb{x}-\Mb{y}))| \  \nu(\mathrm{d}\Mb{x})\  \nu(\mathrm{d}\Mb{y})
\\& = \int_{A} \left [ \int_{A-\Mb{y}} |\mathrm{Cov}(C(\Mb{0}), C(\Mb{z}))| \  \nu(\mathrm{d}\Mb{z}) \right] \  \nu(\mathrm{d}\Mb{y})
\\& \leq \int_{A} \left [ \int_{\mathbb{R}^2} |\mathrm{Cov}(C(\Mb{0}), C(\Mb{z}))| \  \nu(\mathrm{d}\Mb{z}) \right] \  \nu(\mathrm{d}\Mb{y})
\\& = \nu(A) \tilde{\sigma}_C < \infty,
\end{align*}
where 
$$\tilde{\sigma}_C=\int_{\mathbb{R}^2} | \mathrm{Cov}(C(\Mb{0}), C(\Mb{x})) | \  \nu(\mathrm{d}\Mb{x}).$$
Thus, \eqref{Eq_Condition_Existence_Expectation_Product} is obviously satisfied.
\end{proof}

We recall that for a random field $\{ C(\Mb{x}) \}_{\Mb{x} \in \mathbb{R}^2}$ such that, for all $\Mb{x} \in \mathbb{R}^2$, $\mathbb{E} \left[ [C(\Mb{x})]^2 \right]< \infty$, we note
$$\sigma_{C}=\left( \int_{\mathbb{R}^2} \mathrm{Cov}(C(\Mb{0}), C(\Mb{x})) \ \nu(\mathrm{d}\Mb{x}) \right)^{\frac{1}{2}}.$$
Next theorem provides the main results of this subsection. In particular, it gives sufficient conditions on the cost field $C$ such that the induced spatial risk measures $\mathcal{R}_{2}(\cdot, C)$, $\mathcal{R}_{3, \alpha}(\cdot, C)$ and $\mathcal{R}_{4, \alpha}(\cdot, C)$ satisfy the axioms of asymptotic spatial homogeneity of order $-2$, $-1$ and $-1$, respectively.
\begin{Th}
\label{Th_Multiple_Results_General_CP}
Let $\{ C(\Mb{x}) \}_{\Mb{x} \in \mathbb{R}^2} \in \mathcal{C}$.
\Ben
\item Assume that $C$ is stationary. Then, provided it exists, any spatial risk measure associated with a law-invariant classical risk measure $\Pi$ and induced by $C$ satisfies the axiom of spatial invariance under translation.
\item Assume that $C$ is such that, for all $\Mb{x} \in \mathbb{R}^2$, 
\Beq
\label{Eq_Finite_Variance_CP}
\mathbb{E} \left[ [C(\Mb{x})]^2 \right]< \infty,
\Eeq 
and satisfies \eqref{Eq_Dependence_Covariance_x-y} and \eqref{Eq_Convergence_Abs_Int_Cov_R2}.
Then, we have, for all $A \in \mathcal{A}_c$, that
\Beq
\label{Eq_Expression_R2_CP_General}
\mathcal{R}_2(\lambda A, C) \underset{\lambda \to \infty}{=} \dfrac{\sigma_{C}^2}{\lambda^2 \nu(A)}+o \left( \dfrac{1}{\lambda^2} \right).
\Eeq
Hence, if $\sigma_{C} >0$, $\mathcal{R}_2(\cdot, C)$ satisfies the axiom of asymptotic spatial homogeneity of order $-2$ with $K_1(A, C)=0$ and $K_2(A, C)=\sigma_{C}^2/\nu(A), A \in \mathcal{A}_c$.
\item Assume that $C$ has a constant expectation and satisfies the CLT. Then, we have, for all $A \in \mathcal{A}_c$, that
\Beq
\label{Eq_Asymptotics_R3}
\mathcal{R}_{3, \alpha}(\lambda A, C) \underset{\lambda \to \infty}{=} \mathbb{E}[C(\Mb{0})] + \frac{\sigma_{C} q_{\alpha}}{\lambda [\nu(A)]^{\frac{1}{2}}} +o \left( \frac{1}{\lambda} \right).
\Eeq
Hence, if $\alpha \in (0,1) \backslash \{ 1/2 \}$, $\mathcal{R}_{3,\alpha}(\cdot, C)$ satisfies the axiom of asymptotic spatial homogeneity of order $-1$ with $K_1(A, C)=\mathbb{E}[C(\Mb{0})]$ and $K_2(A,C)=\sigma_{C} q_{\alpha}/[\nu(A)]^{\frac{1}{2}}, A \in \mathcal{A}_c$.
\item Assume that $C$ has a constant expectation, satisfies the CLT and is such that the random variables $\lambda \left( L_N(\lambda A, C)-\mathbb{E}[C(\Mb{0})] \right)$, $\lambda  > 0$, are uniformly integrable. Then, we have, for all $A \in \mathcal{A}_c$, that
\Beq
\label{Eq_Asymptotics_R4}
\mathcal{R}_{4, \alpha}(\lambda A, C) \underset{\lambda \to \infty}{=} \mathbb{E}[C(\Mb{0})] + \frac{\sigma_{C} \phi(q_{\alpha})}{\lambda [\nu(A)]^{\frac{1}{2}}(1-\alpha)} +o \left( \frac{1}{\lambda} \right).
\Eeq
Hence, $\mathcal{R}_{4,\alpha}(\cdot, C)$ satisfies the axiom of asymptotic spatial homogeneity of order $-1$ with $K_1(A, C)=\mathbb{E}[C(\Mb{0})]$ and $K_2(A,C)=\sigma_{C} \phi(q_{\alpha})/ \{ [\nu(A)]^{\frac{1}{2}}(1-\alpha) \}$, $A \in \mathcal{A}_c$.
\Een 
\end{Th}

\begin{proof}
1. Let $A \in \mathcal{A}$, $\Mb{v} \in \mathbb{R}^2$ and $\Pi$ be a classical risk measure.
Using the fact that $\nu(A+\Mb{v})=\nu(A)$ and a change of variable, we obtain
\Beq
\label{Eq_Prop_Translation_1}
\mathcal{R}_{\Pi}(A+\Mb{v},C)=\Pi \left( \frac{1}{\nu(A+\Mb{v})} \int_{A+\Mb{v}} C(\Mb{x}) \ \nu(\mathrm{d}\Mb{x}) \right)=\Pi \left( \frac{1}{\nu(A)} \int_{A} C(\Mb{y}+\Mb{v}) \ \nu(\mathrm{d}\Mb{y}) \right).
\Eeq
Due to the stationarity of $C$, we have, for all $\Mb{v} \in \mathbb{R}^2$, that $\{ C(\Mb{x}) \}_{\Mb{x} \in \mathbb{R}^2} \overset{d}{=} \{ C(\Mb{x}+\Mb{v}) \}_{\Mb{x} \in \mathbb{R}^2}$, yielding, since $\Pi$ is law-invariant,
\Beq
\label{Eq_Prop_Translation_2}
\Pi \left( \frac{1}{\nu(A)} \int_{A} C(\Mb{y}+\Mb{v}) \ \nu(\mathrm{d}\Mb{y}) \right) = \Pi \left( \frac{1}{\nu(A)} \int_{A} C(\Mb{x}) \ \nu(\mathrm{d}\Mb{x}) \right)=\mathcal{R}_{\Pi}(A,C).
\Eeq
The combination of \eqref{Eq_Prop_Translation_1} and \eqref{Eq_Prop_Translation_2} provides the result.

\medskip

2. As \eqref{Eq_Dependence_Covariance_x-y} and \eqref{Eq_Convergence_Abs_Int_Cov_R2} are satisfied, we know from Theorem \ref{Th_R2_General} that, for all $A \in \mathcal{A}_c$ and $\lambda >0$, $\mathcal{R}_2(\lambda A, C)$ is well-defined. The result follows from an adapted version of the proof of Theorem 3, Point 3, in \cite{koch2017spatial}. We refer the reader to this proof for the technical parts. We only highlight some of the main steps as well as the main differences here.

The first part consists in showing that
\Beq
\label{Eq_To_Obtain_Asymp_Spat_Homog_R2}
\lim_{\lambda \to \infty} \lambda^2 \nu(A) \mathcal{R}_2 (\lambda A, C) = \sigma_{C}^2.
\Eeq
Using Theorem \ref{Th_R2_General} and \eqref{Eq_Dependence_Covariance_x-y}, it follows that
$$\mathcal{R}_2(\lambda A, C)= \frac{1}{\lambda^4 [\nu(A)]^2} \int_{\lambda A}  \int_{\lambda A} \mathrm{Cov}(C(\Mb{0}), C(\Mb{x}-\Mb{y})) \  \nu(\mathrm{d}\Mb{x})\  \nu(\mathrm{d}\Mb{y}).$$
Let $A_{\lambda}=\lambda A, \lambda>0$. Then, we consider the quantity
$$T_{\lambda}=\frac{1}{\lambda^2 \nu(A)} \int_{A_{\lambda}} \int_{A_{\lambda}} k(\Mb{x}-\Mb{y})\  \  \nu(\mathrm{d}\Mb{x})\  \nu(\mathrm{d}\Mb{y}), \quad \lambda >0,$$
where
$$ k(\Mb{x})=\mathrm{Cov}(C(\Mb{0}), C(\Mb{x})), \quad \Mb{x} \in \mathbb{R}^2.$$
The next step consists in showing that 
\Beq
\label{Eq_Lim_Tlambda}
\lim_{\lambda \to \infty} T_{\lambda}= \sigma_{C}^2. 
\Eeq
For this purpose, we proceed similarly as in \cite{koch2017spatial}, proof of Theorem 3, Point 3.
The only difference consists in the fact that here $k$ is not necessarily non-negative. Hence, in order to bound $| T_{1, \lambda} |$ and $| T_{3, \lambda} |$ (these quantities are defined in \cite{koch2017spatial}) from above, $k$ must be replaced with its absolute value in the corresponding integrals. This is where Condition \eqref{Eq_Convergence_Abs_Int_Cov_R2} plays a role.
Finally, since, for all $\lambda>0$,
$$ T_{\lambda}=\lambda^2 \nu(A) \mathcal{R}_2 (\lambda A, C),$$
\eqref{Eq_To_Obtain_Asymp_Spat_Homog_R2} follows from \eqref{Eq_Lim_Tlambda}.

In a second part, we easily derive \eqref{Eq_Expression_R2_CP_General} from \eqref{Eq_To_Obtain_Asymp_Spat_Homog_R2}.
Now, as a compact subset of $\mathbb{R}^2$, $A$ is bounded, giving that $\nu(A) \in (0, \infty)$.
Since, moreover, $\sigma_{C}^2 \in (0, \infty)$, $\sigma_{C}^2/\nu(A) \in (0, \infty)$. Hence, the second part of the result follows from \eqref{Eq_Expression_R2_CP_General}.

\medskip

3. Theorem \ref{Th_Link_CLT_VanHove_CLT_Homothety} gives that, for all $A \in \mathcal{A}_c$,
$$\lambda \left( L_N(\lambda A, C) - \mathbb{E}[C(\Mb{0})] \right) \overset{d}{\to} \mathcal{N} \left( 0, \frac{\sigma_{C}^2}{\nu(A)} \right), \mbox{ for } \lambda \to \infty.$$
Hence, the fact that the quantile function of a normal random variable is continuous on $(0, 1)$, Proposition 0.1 in \cite{resnickextreme} and easy computations \citep[see][proof of Theorem 5]{koch2017spatial} yield \eqref{Eq_Asymptotics_R3}. Since $C$ satisfies the CLT, we have $\mathbb{E} \left[ [C(\Mb{0})]^2 \right]< \infty$ and thus $\mathbb{E} \left[ C(\Mb{0}) \right]< \infty$. Additionally, as $\alpha \neq 0.5$, we have $q_{\alpha} \neq 0$. Moreover, as $\sigma_{C}>0$ (because $C$ satisfies the CLT) and $\nu(A)>0$, we obtain that $\sigma_{C} q_{\alpha}/[\nu(A)]^{\frac{1}{2}} \neq 0$. Finally, since $\alpha \notin \{0, 1 \}$, $| q_{\alpha}| < \infty$. Furthermore, $\sigma_{C} < \infty$ and $\nu(A) < \infty$, giving that $| \sigma_{C} q_{\alpha}/[\nu(A)]^{\frac{1}{2}}| < \infty$. The result follows by definition.

\medskip

4. Since $C$ satisfies the CLT, we have, for all $\Mb{x} \in \mathbb{R}^2$, $\mathbb{E} \left[ [C(\Mb{x})]^2 \right]< \infty$, which implies that, for all $\Mb{x} \in \mathbb{R}^2$, $\mathbb{E} \left[ | C(\Mb{x}) | \right]< \infty$. We easily deduce, using Fubini's theorem, that $\mathbb{E}[|L_N(\lambda A, C)|]$ is finite, and, therefore, that $\mathcal{R}_{4,\alpha}(\lambda A, C)$ is well-defined for all $A \in \mathcal{A}_c$ and $\lambda >0$. Theorem \ref{Th_Link_CLT_VanHove_CLT_Homothety} gives that, for all $A \in \mathcal{A}_c$,
$$\lambda \left( L_N(\lambda A, C) - \mathbb{E}[C(\Mb{0})] \right) \overset{d}{\to} \mathcal{N} \left( 0, \frac{\sigma_{C}^2}{\nu(A)} \right), \mbox{ for } \lambda \to \infty.$$
Now, ES is known to be continuous with respect to convergence in distribution in the case of uniformly integrable random variables. For details, we refer for instance to \cite{wang2018characterization}, Theorem 3.2 and Example 2.2, Point (ii); the authors' results concern bounded random variables but the mentioned result can be extended to the case of integrable random variables. Hence, it follows from the fact that the random variables $\lambda \left( L_N(\lambda, C)-\mathbb{E}[C(\Mb{0})] \right)$, $\lambda  > 0$, are uniformly integrable, and the expression of $\mathrm{ES}_{\alpha}$ for the Gaussian distribution, that
\Beq
\label{Eq_Convergence_TVaR}
\lim_{\lambda \to \infty} \frac{1}{1-\alpha} \int_{\alpha}^{1} \mbox{VaR}_{u} ( \lambda [L_N(\lambda A, C)-\mathbb{E}[C(\Mb{0})]] ) \ \nu(\mathrm{d}u)=\frac{\sigma_{C} \phi(q_{\alpha})}{[\nu(A)]^{\frac{1}{2}}(1-\alpha)}.
\Eeq
Moreover, we have
\begin{align*}
& \quad \  \frac{1}{1-\alpha} \int_{\alpha}^{1} \mbox{VaR}_{u} ( \lambda [L_N(\lambda A, C)-\mathbb{E}[C(\Mb{0})]] ) \ \nu(\mathrm{d}u) 
\\& = \frac{1}{1-\alpha} \int_{\alpha}^{1} \lambda \left( \mbox{VaR}_{u} ( L_N(\lambda A, C))-\mathbb{E}[C(\Mb{0})] \right) \ \nu(\mathrm{d}u)
\\& =\lambda \left( \mathcal{R}_{4,\alpha}(\lambda A, C)-\mathbb{E}[C(\Mb{0})] \right).
\end{align*}
Thus, \eqref{Eq_Convergence_TVaR} gives, for all $A \in \mathcal{A}_c$,
$$\lambda \left( \mathcal{R}_{4,\alpha}(\lambda A, C)-\mathbb{E}[C(\Mb{0})] \right) \underset{\lambda \to \infty}{=} \frac{\sigma_{C} \phi(q_{\alpha})}{[\nu(A)]^{\frac{1}{2}}(1-\alpha)} + o(1),$$
which yields \eqref{Eq_Asymptotics_R4}. Now, we have $\mathbb{E} \left[ C(\Mb{0}) \right]< \infty$.
Moreover, using the fact that, for all $\alpha \in (0,1)$, $\phi(q_{\alpha}) \in (0, \infty)$, and arguments stated at the end of the proof of Point 3, we obtain $| \sigma_{C} \phi(q_{\alpha})/ \{ [\nu(A)]^{\frac{1}{2}}(1-\alpha) \} | \in (0, \infty)$. Consequently, the result follows by definition.
\end{proof}
\begin{Rq}
In  order to establish Points 3 and 4, we take advantage of the fact that both VaR and ES are continuous with respect to convergence in distribution under appropriate assumptions. Hence, similar results might hold for other classical risk measures satisfying continuity with respect to convergence in distribution.
\end{Rq}

Theorem \ref{Th_Multiple_Results_General_CP} entails the following important result.
\begin{Corr}
\label{Corr_Implication_CLT_ASH_R2}
Let $\{ C(\Mb{x}) \}_{\Mb{x} \in \mathbb{R}^2} \in \mathcal{C}$. Moreover, assume that $C$ satisfies \eqref{Eq_Dependence_Covariance_x-y} and the CLT. Then, we have, for all $A \in \mathcal{A}_c$, that
\Beq
\label{Eq_Implication_CLT_ASH_R2}
\mathcal{R}_2(\lambda A, C) \underset{\lambda \to \infty}{=} \dfrac{\sigma_{C}^2}{\lambda^2 \nu(A)}+o \left( \dfrac{1}{\lambda^2} \right).
\Eeq
Hence, $\mathcal{R}_2(\cdot, C)$ satisfies the axiom of asymptotic spatial homogeneity of order $-2$ with $K_1(A, C)=0$ and $K_2(A, C)=\sigma_{C}^2/\nu(A), A \in \mathcal{A}_c$.
\end{Corr}
\begin{proof}
Since $C$ satisfies the CLT, it satisfies \eqref{Eq_Convergence_Abs_Int_Cov_R2}, \eqref{Eq_Finite_Variance_CP} and $\sigma_C>0$. Thus, the result follows from Theorem \ref{Th_Multiple_Results_General_CP}, Point 2.
\end{proof}

Next result provides a convenient condition ensuring the uniform integrability required in Theorem \ref{Th_Multiple_Results_General_CP}, Point 4.  
\begin{Prop}
\label{Prop_Condition_Uniform_Integrability}
Let $\{ C(\Mb{x}) \}_{\Mb{x} \in \mathbb{R}^2} \in \mathcal{C}$. Assume moreover that $C$ has a constant expectation and satisfies the CLT. If $C$ satisfies \eqref{Eq_Dependence_Covariance_x-y}, then the random variables 
$\lambda \left( L_N(\lambda A, C)-\mathbb{E}[C(\Mb{0})] \right)$, $\lambda  > 0$, are uniformly integrable.
\end{Prop}
\begin{proof}
Let, for $\lambda>0$, $M_{\lambda}= \lambda \left( L_N(\lambda A, C) - \mathbb{E}[C(\Mb{0})] \right)$. Theorem \ref{Th_Link_CLT_VanHove_CLT_Homothety} gives that, for all $A \in \mathcal{A}_c$,
$M_{\lambda} \overset{d}{\to} M, \mbox{ for } \lambda \to \infty$, where $M \sim \mathcal{N} \left( 0, \sigma_{C}^2/\nu(A) \right)$. Therefore, by the continuous mapping theorem, we obtain
\Beq
\label{Eq_Convergence_Square}
M_{\lambda}^2 \overset{d}{\to} M^2, \mbox{ for } \lambda \to \infty.
\Eeq
Now, it is clear that, for all $\lambda >0$, $\mathrm{Var}(M_{\lambda})=\lambda^2 \mathcal{R}_2(\lambda A, C)$. Hence, it follows from \eqref{Eq_Implication_CLT_ASH_R2} that 
$\mathrm{Var}(M_{\lambda}) \underset{\lambda \to \infty}{\to} \sigma_{C}^2/\nu(A)$, which gives, since for all $\lambda >0$, $\mathbb{E}[M_{\lambda}] = 0$, that $\mathbb{E} \left[ M_{\lambda}^2 \right] \underset{\lambda \to \infty}{\to} \mathbb{E} \left[ M^2 \right]$.
Additionally, $M^2$ is non-negative and integrable. Furthermore, the $M_{\lambda}^2$ are non-negative and, for all $\lambda>0$, $\mathbb{E} \left[ M_{\lambda}^2 \right]=\lambda^2 \mathcal{R}_2(\lambda A, C)$, which is finite according to Theorem \ref{Th_R2_General} as \eqref{Eq_Condition_Existence_Expectation_Product} is satisfied. Therefore, the $M_{\lambda}^2$ are integrable. Consequently, using \eqref{Eq_Convergence_Square} and Theorem 3.6 in \cite{billingsley1999convergence}, we know that the random variables $M_{\lambda}^2$, $\lambda>0$, are uniformly integrable. This directly yields that the random variables $M_{\lambda}$, $\lambda>0$, are uniformly integrable.
\end{proof}

\subsection{Cost field being a function of a max-stable random field}
\label{Subsec_Costfield_Function_Maxstable}

We now consider a cost field model written as in \eqref{Eq_Economic_Loss_Model}, i.e.,
\Beq
\label{Eq_Economic_Loss_Model_Max_Stable}
\left \{ C(\Mb{x}) \right \}_{\Mb{x} \in \mathbb{R}^2} = \left \{ E(\Mb{x})\  D \left( Z(\Mb{x}) \right) \right \}_{\Mb{x} \in \mathbb{R}^2},
\Eeq
where $Z$ is max-stable and the exposure is uniformly equal to unity. The relevance of using max-stable random fields has been previously highlighted. 

In the following, all theorems and corollaries assume $Z$ to be simple, although max-stable fields fitted to real data have generalized extreme-value (GEV) univariate marginal distributions with location, scale and shape parameters $\eta \in \mathbb{R}$, $\tau >0$ and $\xi \in \mathbb{R}$. However, this does not cause any loss of generality. If $\{ Z(\Mb{x}) \}_{\Mb{x} \in \mathbb{R}^2}$ is a max-stable field with such GEV parameters, we can write
\Beq
Z(\Mb{x}) = 
\left \{
\begin{array}{ll}
\eta + \tau (\tilde{Z}(\Mb{x})^{\xi}-1)/\xi & \mbox{if} \quad \xi \neq 0, \\
\eta + \tau \log( \tilde{Z}(\Mb{x})) & \mbox{if} \quad \xi=0,
\end{array}
\quad \Mb{x} \in \mathbb{R}^2,
\right.
\Eeq 
where $\{ \tilde{Z}(\Mb{x}) \}_{\Mb{x} \in \mathbb{R}^2}$ is simple max-stable. Thus, there exists a function $D_1$ such that $Z(\Mb{x})=D_1(\tilde{Z}(\Mb{x}))$ and Model \eqref{Eq_Economic_Loss_Model_Max_Stable} can be written $C(\Mb{x}) = \tilde{D}(\tilde{Z}(\Mb{x}))$, where $\tilde{Z}$ is simple max-stable and $\tilde{D} = D \circ D_1$, with ``$\circ$'' denoting function composition. 
On $(0, \infty)$, for any $\xi \neq 0$ the transformation $\tilde{z} \mapsto \eta + \tau (\tilde{z}^{\xi}-1)/\xi$ is increasing and the same holds for the transformation $\tilde{z} \mapsto \eta + \tau \log(\tilde{z})$, implying that $D_1$ is increasing. Most often, the damage function $D$ is also increasing (e.g, the higher the wind speed, temperature or rainfall amount, the higher the cost) and thus the same is true for $\tilde{D}=D \circ D_1$. Consequently, the requirement in Corollaries \ref{Corr_Sigma2_Positive}--\ref{Corr_R3_Orderminus1_BR_Generalvariogram} (see below) on the function applied to the simple max-stable field to be non-decreasing and non-constant is generally satisfied in the applications motivating the present work.

For the sake of notational simplicity, in the following, we denote by $Z$ (instead of $\tilde{Z}$) the simple max-stable field and by $D$ (instead of $\tilde{D}$) the quantity $D \circ D_1$. Accordingly, the reader should pay attention to the fact that $Z$ models the standardized environmental field (and not the real one) and $D$ consists in the composition of the marginal transformation of $Z$ and the damage function.

We first give sufficient conditions on the function $D$ and the field $Z$ such that the spatial risk measure $\mathcal{R}_1(\cdot, D(Z))$ induced by the cost field $D(Z)$ satisfies the axioms presented in Definition \ref{Chapriskmeasures_Def_Axiomatic_Risk_Measure_Function_Cp}.
\begin{Corr}
Let $\{ Z(\Mb{x}) \}_{\Mb{x} \in \mathbb{R}^2}$ be a simple max-stable random field and $D$ a measurable function such that $\{ C(\Mb{x}) \}_{\Mb{x} \in \mathbb{R}^2}=\{ D(Z(\Mb{x})) \}_{\Mb{x} \in \mathbb{R}^2} \in \mathcal{C}$ and $\mathbb{E}[|C(\Mb{0})|] < \infty$. Then, for all $A \in \mathcal{A}$, $\mathcal{R}_1(A, C)=\mathbb{E}[C(\Mb{0})]$. Hence, $\mathcal{R}_1(\cdot, C)$ satisfies the axioms of spatial invariance under translation and spatial sub-additivity. If, moreover, $\mathbb{E}[C(\Mb{0})] \neq 0$, then $\mathcal{R}_1(\cdot, C)$ satisfies the axiom of asymptotic spatial homogeneity of order $0$ with $K_1(A, C)=0$ and $K_2(A, C)=\mathbb{E}[C(\Mb{0})], A \in \mathcal{A}_c$.
\end{Corr}
\begin{proof}
Since $Z$ has identical margins, for all $\Mb{x} \in \mathbb{R}^2$, $\mathbb{E} \left[ |C(\Mb{x})| \right]=\mathbb{E} \left[ |C(\Mb{0})| \right]$. Therefore, the result directly follows from Theorem \ref{Th_R1_General_Cost_Field}.
\end{proof}

The result below gives sufficient conditions on $D$ and $Z$ such that the spatial risk measure $\mathcal{R}_2(\cdot, D(Z))$ induced by the cost field $D(Z)$ satisfies the axiom of asymptotic spatial homogeneity of order $-2$.
\begin{Th}
Let $\{ Z(\Mb{x}) \}_{\Mb{x} \in \mathbb{R}^2}$ be a simple and sample-continuous max-stable random field and $D$ a measurable function such that 
$\{ C(\Mb{x}) \}_{\Mb{x} \in \mathbb{R}^2}=\{ D(Z(\Mb{x})) \}_{\Mb{x} \in \mathbb{R}^2} \in \mathcal{C}$ and such that there exist $p, q >0$ satisfying $2/p+1/q=1$ such that
\Beq
\label{Eq_Condition_Expectation}
\mathbb{E} \left[ |C(\Mb{0})|^p \right] < \infty
\Eeq
and 
\Beq
\label{Eq_Condition_Theta}
\int_{\mathbb{R}^2} [2-\theta(\Mb{0}, \Mb{x})]^{\frac{1}{q}} \  \nu(\mathrm{d}\Mb{x}) < \infty,
\Eeq 
where $\theta$ is the extremal coefficient function of $Z$. Then, we have
$$ \int_{\mathbb{R}^2} | \mathrm{Cov}(C(\Mb{0}), C(\Mb{x}))| \ \nu(\mathrm{d}\Mb{x}) < \infty.$$
Additionally, assume that $C$ satisfies \eqref{Eq_Dependence_Covariance_x-y} and $\sigma_{C}>0$. Then $\mathcal{R}_2(\cdot, C)$ satisfies the axiom of asymptotic spatial homogeneity of order $-2$ with $K_1(A,C)=0$ and $K_2(A,C)=\sigma_{C}^2/\nu(A)$, $A \in \mathcal{A}_c$.
\end{Th}

\begin{proof}
Since $Z$ has identical margins, \eqref{Eq_Condition_Expectation} yields that, for all $\Mb{x} \in \mathbb{R}^2$, $\mathbb{E} \left[ |C(\Mb{x})|^p \right] < \infty$. Thus, using the fact that $2/p+1/q=1$, Davydov's inequality \citep[][Equation (2.2)]{davydov1968convergence} gives that 
\Beq
\label{Eq_Davydov_Inequality}
| \mathrm{Cov}(C(\Mb{0}), C(\Mb{x})) | \leq 12 \left[ \alpha^{C}(\{ \Mb{0} \}, \{ \Mb{x} \}) \right]^{\frac{1}{q}} \left( \mathbb{E} \left[ |C(\Mb{0})|^{p} \right] \right)^{\frac{1}{p}} \left( \mathbb{E} \left[ |C(\Mb{x})|^{p} \right] \right)^{\frac{1}{p}}.
\Eeq
For all $\Mb{x} \in \mathbb{R}^2$, since $D$ is measurable, $C(\Mb{x})=D(Z(\Mb{x}))$ is $\mathcal{F}^{Z}_{ \{ \Mb{x} \} }$-measurable. Hence, $\mathcal{F}^{C}_{ \{ \Mb{x} \} } \subset \mathcal{F}^{Z}_{ \{ \Mb{x} \} }$, which gives by \eqref{Eq_Def_Alpha_Mixing_Coefficient} that, for all $\Mb{x} \in \mathbb{R}^2$,
\Beq
\label{Eq_Majoration_Alpha_Mixing_X_With_Alpha_Mixing_Z}
\alpha^{C} \left( \{ \Mb{0} \}, \{ \Mb{x} \} \right) \leq \alpha^Z \left( \{ \Mb{0} \}, \{ \Mb{x} \} \right).
\Eeq
Now, using \eqref{Eq_Majoration_Alphamixing_With_Betamixing} and Corollary 2.2 in \cite{dombry2012strong}, we obtain that, for all $\Mb{x} \in \mathbb{R}^2$,
\Beq
\label{Eq_Majoration_Alpha_Mixing_Z}
\alpha^Z \left( \{ \Mb{0} \}, \{ \Mb{x} \} \right) \leq 2 [2-\theta(\Mb{0}, \Mb{x})].
\Eeq
Thus, the combination of \eqref{Eq_Majoration_Alpha_Mixing_X_With_Alpha_Mixing_Z} and \eqref{Eq_Majoration_Alpha_Mixing_Z} gives that
$$ \alpha^{C} \left( \{ \Mb{0} \}, \{ \Mb{x} \} \right) \leq 2 [2-\theta(\Mb{0}, \Mb{x})].$$
Consequently, \eqref{Eq_Davydov_Inequality} gives that
$$ | \mathrm{Cov}(C(\Mb{0}), C(\Mb{x})) | \leq 12 \ 2^{\frac{1}{q}} \left( \mathbb{E} \left[ |C(\Mb{0})|^{p} \right] \mathbb{E} \left[ |C(\Mb{x})|^{p} \right] \right)^{\frac{1}{p}} [2-\theta(\Mb{0}, \Mb{x})]^{\frac{1}{q}}.$$
Therefore, using \eqref{Eq_Condition_Expectation} and \eqref{Eq_Condition_Theta}, we obtain
$$\displaystyle \int_{\mathbb{R}^2} | \mathrm{Cov}(C(\Mb{0}), C(\Mb{x}))| \ \nu(\mathrm{d}\Mb{x}) < \infty.$$
Since $p,q >0$ and $2/p+1/q=1$, we have $p>2$. Consequently, for all $\Mb{x} \in \mathbb{R}^2$, $\mathbb{E} \left[ [C(\Mb{x})]^2 \right]<\infty$. Thus, Theorem \ref{Th_Multiple_Results_General_CP}, Point 2, gives the result.
\end{proof}

Until the end, the following results provide sufficient conditions on $D$ and $Z$ such that the induced spatial risk measures $\mathcal{R}_{2}(\cdot, D(Z))$, $\mathcal{R}_{3, \alpha}(\cdot, D(Z))$ and $\mathcal{R}_{4, \alpha}(\cdot, D(Z))$ satisfy the axioms of asymptotic spatial homogeneity of order $-2$, $-1$ and $-1$, respectively. In order to establish them, we take advantage of the results in \cite{koch2017TCL} about the existence of a CLT for functions of stationary max-stable random fields. Let $\mathcal{B}(\mathbb{R})$ and $\mathcal{B}((0, \infty))$ be the Borel $\sigma$-fields on $\mathbb{R}$ and $(0, \infty)$, respectively. For $\Mb{h}=(h_1, h_2)^{\prime} \in \mathbb{Z}^2$, we adopt the notation $[\Mb{h}, \Mb{h}+1]=[h_1,h_1+1] \times [h_2, h_2+1]$. Next theorem considers a general simple, stationary and sample-continuous max-stable random field.
\begin{Th}
\label{Th_R3_Max_Stable}
Let $\{ Z(\Mb{x}) \}_{\Mb{x} \in \mathbb{R}^2}$ be a simple, stationary and sample-continuous max-stable random field and $D$ be a measurable function from $((0, \infty), \mathcal{B}((0,\infty)))$ to $(\mathbb{R},\mathcal{B}(\mathbb{R}))$ satisfying
\Beq
\label{Eq_Assumption_F}
\mathbb{E}\left[ |D(Z(\Mb{0}))|^{2+\delta} \right]<\infty,
\Eeq
for some $\delta>0$. Furthermore, assume that, for all $\Mb{h} \in \mathbb{Z}^2$, 
$$
\mathbb{E} \left[ \min \left \{ \sup_{\Mb{x} \in [0,1]^2} \{ Y(\Mb{x}) \},   \sup_{\Mb{x} \in [\Mb{h},\Mb{h}+1]}  \{ Y(\Mb{x}) \} \right \} \right] \leq K \| \Mb{h} \|^{-b},
$$
for some $K>0$, $b > 2 \max  \left \{ 2, (2+\delta)/\delta \right \}$ and where $\{ Y(\Mb{x}) \}_{\Mb{x} \in \mathbb{R}^2}$ is a spectral random field of $Z$ (see \eqref{Eq_Spectral_Representation_Stochastic_Processes}).
Let $\{ C(\Mb{x}) \}_{\Mb{x} \in \mathbb{R}^2}= \{  D(Z(\Mb{x})) \}_{\Mb{x} \in \mathbb{R}^2}$.
Then, if $\sigma_{C}>0$:
\Ben
\item $\mathcal{R}_2(\cdot, C)$ satisfies the axiom of asymptotic spatial homogeneity of order $-2$ with $K_1(A, C)=0$ and $K_2(A, C)=\sigma_{C}^2/\nu(A), A \in \mathcal{A}_c$.
\item For all $\alpha \in (0,1) \backslash \{ 1/2 \}$, $\mathcal{R}_{3,\alpha}(\cdot, C)$ satisfies the axiom of asymptotic spatial homogeneity of order $-1$ with $K_1(A,C)=\mathbb{E}[C(\Mb{0})]$ and $K_2(A,C)=\sigma_{C} q_{\alpha}/[\nu(A)]^{\frac{1}{2}}$, $A \in \mathcal{A}_c$.
\item For all $\alpha \in (0,1)$, $\mathcal{R}_{4,\alpha}(\cdot, C)$ satisfies the axiom of asymptotic spatial homogeneity of order $-1$ with $K_1(A,C)=\mathbb{E}[C(\Mb{0})]$ and $K_2(A,C)=\sigma_{C} \phi(q_{\alpha})/ \{ [\nu(A)]^{\frac{1}{2}}(1-\alpha) \}$, $A \in \mathcal{A}_c$.
\Een
\end{Th}

\begin{proof}
Since $Z$ is sample-continuous, it is measurable. Thus, the function $D$ being measurable from $((0, \infty), \mathcal{B}((0,\infty)))$ to $(\mathbb{R},\mathcal{B}(\mathbb{R}))$, we obtain that $C$ is measurable. Moreover, it follows from the stationarity of $C$ (due to the stationarity of $Z$) and Condition \eqref{Eq_Assumption_F} that, for all $\Mb{x} \in \mathbb{R}^2$, $\mathbb{E}\left[ |C(\Mb{x})| \right]=\mathbb{E}\left[ |C(\Mb{0})| \right]<\infty$. Therefore, the function $\Mb{x} \mapsto \mathbb{E}[|C(\Mb{x})|]$ is constant and hence obviously locally integrable. Consequently, Proposition \ref{Eq_Sufficient_Condition_Locally_Integrable_Sample_Paths} gives that $C$ has a.s. locally integrable sample paths. Therefore, $C \in \mathcal{C}$.

Furthermore, the assumptions enable us to apply Theorem 2 in \cite{koch2017TCL}. The latter yields that the random field $C$ satisfies the CLT. Finally, since $C$ is stationary, it satisfies \eqref{Eq_Dependence_Covariance_x-y} and has a constant expectation. Hence, Corollary \ref{Corr_Implication_CLT_ASH_R2} gives the first result. The second result follows from Theorem \ref{Th_Multiple_Results_General_CP}, Point 3. The combination of Proposition \ref{Prop_Condition_Uniform_Integrability} and Point 4 in Theorem \ref{Th_Multiple_Results_General_CP} yields the third result.
\end{proof}

Theorem \ref{Th_R3_Max_Stable} directly entails the following result.
\begin{Corr}
\label{Corr_Sigma2_Positive}
Let $Z$, $D$ and $C$ be as in Theorem \ref{Th_R3_Max_Stable} (but without assuming that $\sigma_C>0$). Moreover, assume that $D$ is non-decreasing and non-constant. 
Then:
\Ben
\item $\mathcal{R}_2(\cdot, C)$ satisfies the axiom of asymptotic spatial homogeneity of order $-2$ with $K_1(A, C)=0$ and $K_2(A, C)=\sigma_{C}^2/\nu(A), A \in \mathcal{A}_c$.
\item For all $\alpha \in (0,1) \backslash \{ 1/2 \}$, $\mathcal{R}_{3,\alpha}(\cdot, C)$ satisfies the axiom of asymptotic spatial homogeneity of order $-1$ with $K_1(A,C)=\mathbb{E}[C(\Mb{0})]$ and $K_2(A,C)=\sigma_{C} q_{\alpha}/[\nu(A)]^{\frac{1}{2}}$, $A \in \mathcal{A}_c$.
\item For all $\alpha \in (0,1)$, $\mathcal{R}_{4,\alpha}(\cdot, C)$ satisfies the axiom of asymptotic spatial homogeneity of order $-1$ with $K_1(A,C)=\mathbb{E}[C(\Mb{0})]$ and $K_2(A,C)=\sigma_{C} \phi(q_{\alpha})/ \{ [\nu(A)]^{\frac{1}{2}}(1-\alpha) \}$, $A \in \mathcal{A}_c$.
\Een
\end{Corr}

\begin{proof}
Proposition 1 in \cite{koch2017TCL} gives that $\sigma_{C}>0$. Therefore, Theorem \ref{Th_R3_Max_Stable} yields the result.
\end{proof}

The next results concern the Brown--Resnick and Smith max-stable random fields. The Brown--Resnick model is of high practical interest since, owing to its flexibility, it appears as one of the best (if not the best) models among currently available max-stable models, at least for environmental data; see, e.g., \citet[][Section 7.4]{davison2012statistical}, in the case of rainfall.
\begin{Th}
\label{Th_R3_Brown_Resnick_Power_Variogram}
Let $\{ Z(\Mb{x}) \}_{\Mb{x} \in \mathbb{R}^2}$ be the Brown--Resnick random field associated with the variogram $\gamma_W(\Mb{x})=m \| \Mb{x} \|^{\psi}$, where $m >0$ and $\psi \in (0,2]$, or the Smith random field with covariance matrix $\Sigma$, and $D$ be as in Theorem \ref{Th_R3_Max_Stable}. Let $\{ C(\Mb{x}) \}_{\Mb{x} \in \mathbb{R}^2}= \{ D(Z(\Mb{x})) \}_{\Mb{x} \in \mathbb{R}^2}$. Then, if $\sigma_{C}>0$:
\Ben
\item $\mathcal{R}_2(\cdot, C)$ satisfies the axiom of asymptotic spatial homogeneity of order $-2$ with $K_1(A, C)=0$ and $K_2(A, C)=\sigma_{C}^2/\nu(A), A \in \mathcal{A}_c$.
\item For all $\alpha \in (0,1) \backslash \{ 1/2 \}$, $\mathcal{R}_{3,\alpha}(\cdot, C)$ satisfies the axiom of asymptotic spatial homogeneity of order $-1$ with $K_1(A,C)=\mathbb{E}[C(\Mb{0})]$ and $K_2(A,C)=\sigma_{C} q_{\alpha}/[\nu(A)]^{\frac{1}{2}}$, $A \in \mathcal{A}_c$.
\item For all $\alpha \in (0,1)$, $\mathcal{R}_{4,\alpha}(\cdot, C)$ satisfies the axiom of asymptotic spatial homogeneity of order $-1$ with $K_1(A,C)=\mathbb{E}[C(\Mb{0})]$ and $K_2(A,C)=\sigma_{C} \phi(q_{\alpha})/ \{ [\nu(A)]^{\frac{1}{2}}(1-\alpha) \}$, $A \in \mathcal{A}_c$.
\Een
\end{Th}

\begin{proof}
We start with the proof in the case of the Brown--Resnick field.
As previously mentioned, the Brown--Resnick random field is stationary. Thus, $C$ is stationary and hence satisfies \eqref{Eq_Dependence_Covariance_x-y} and has a constant expectation. Moreover, we can see from the proof of Theorem 3 in \cite{koch2017TCL} that $Z$ is sample-continuous. Consequently, the same arguments as in the proof of Theorem \ref{Th_R3_Max_Stable} yield that $C \in \mathcal{C}$. Furthermore, Theorem 3 in \cite{koch2017TCL} gives that $C$ satisfies the CLT. Therefore, Corollary \ref{Corr_Implication_CLT_ASH_R2} yields the first result. The second result follows from Theorem \ref{Th_Multiple_Results_General_CP}, Point 3. The combination of Proposition \ref{Prop_Condition_Uniform_Integrability} and Point 4 in Theorem \ref{Th_Multiple_Results_General_CP} gives the third result.

The Smith random field is stationary as an instance of M3 random field. Thus, $C$ is stationary and consequently satisfies \eqref{Eq_Dependence_Covariance_x-y} and has a constant expectation. Moreover, as the Smith field is sample-continuous, the same arguments as in the proof of Theorem \ref{Th_R3_Max_Stable} yield that $C \in \mathcal{C}$. Additionally, Theorem 4 in \cite{koch2017TCL} gives that $C$ satisfies the CLT. Therefore, Corollary \ref{Corr_Implication_CLT_ASH_R2} yields the first result. The second result follows from Theorem \ref{Th_Multiple_Results_General_CP}, Point 3. The combination of Proposition \ref{Prop_Condition_Uniform_Integrability} and Point 4 in Theorem \ref{Th_Multiple_Results_General_CP} gives the third result.
\end{proof}

Next corollary easily follows from Theorem \ref{Th_R3_Brown_Resnick_Power_Variogram}.
\begin{Corr}
\label{Corr_R3_Orderminus1_BR_Powervariogram}
Let $Z$, $D$ and $C$ be as in Theorem \ref{Th_R3_Brown_Resnick_Power_Variogram} (but without assuming that $\sigma_C>0$). 
Moreover, assume that $D$ is non-decreasing and non-constant. 
Then:
\Ben
\item $\mathcal{R}_2(\cdot, C)$ satisfies the axiom of asymptotic spatial homogeneity of order $-2$ with $K_1(A, C)=0$ and $K_2(A, C)=\sigma_{C}^2/\nu(A), A \in \mathcal{A}_c$.
\item For all $\alpha \in (0,1) \backslash \{ 1/2 \}$, $\mathcal{R}_{3,\alpha}(\cdot, C)$ satisfies the axiom of asymptotic spatial homogeneity of order $-1$ with $K_1(A,C)=\mathbb{E}[C(\Mb{0})]$ and $K_2(A,C)=\sigma_{C} q_{\alpha}/[\nu(A)]^{\frac{1}{2}}$, $A \in \mathcal{A}_c$.
\item For all $\alpha \in (0,1)$, $\mathcal{R}_{4,\alpha}(\cdot, C)$ satisfies the axiom of asymptotic spatial homogeneity of order $-1$ with $K_1(A,C)=\mathbb{E}[C(\Mb{0})]$ and $K_2(A,C)=\sigma_{C} \phi(q_{\alpha})/ \{ [\nu(A)]^{\frac{1}{2}}(1-\alpha) \}$, $A \in \mathcal{A}_c$.
\Een
\end{Corr}

\begin{proof}
As explained in the proof of Theorem \ref{Th_R3_Brown_Resnick_Power_Variogram}, both such Brown--Resnick fields and the Smith field are stationary and sample-continuous. Furthermore, they are simple max-stable. Thus, Proposition 1 in \cite{koch2017TCL} gives that $\sigma_{C}>0$. Hence, Theorem \ref{Th_R3_Brown_Resnick_Power_Variogram} yields the result.
\end{proof}

Let $\| . \|$ denote the Euclidean distance in $\mathbb{R}^2$. We introduce $\mathcal{B}_1= \left \{ \Mb{x} \in \mathbb{R}^2: \| \Mb{x} \|=1 \right \}$, the unit ball of $\mathbb{R}^2$. For two functions $g_1$ and $g_2$ from $\mathbb{R}^2$ to $\mathbb{R}$, the notation $g_1(\Mb{h}) \underset{\| \Mb{h} \| \to \infty}{=} o(g_2(\Mb{h}))$ means that $\lim_{h \to \infty} \sup_{\Mb{u} \in \mathcal{B}_1} \left \{ | g_1(h \Mb{u})/g_2(h \Mb{u})| \right \}=0$. Moreover, $\lim_{\| \Mb{h} \| \to \infty} g_1(\Mb{h})=\infty$ must be understood as $\lim_{h \to \infty} \inf_{\Mb{u} \in \mathcal{B}_1} \left \{ g_1(h \Mb{u}) \right \}=\infty$.
\begin{Th}
\label{Th_R3_Brown_Resnick_General_Variogram}
Let $\{ Z(\Mb{x}) \}_{\Mb{x} \in \mathbb{R}^2}$ be the Brown--Resnick random field built with a random field $\{ W(\Mb{x}) \}_{\Mb{x} \in \mathbb{R}^2}$ which is sample-continuous and whose variogram satisfies
$$\sup_{\Mb{x} \in [0,1]^2} \{ \gamma_W(\Mb{h})-\gamma_W(\Mb{x}+\Mb{h}) \}\underset{\| \Mb{h} \| \to \infty}{=}o(\gamma_W(\Mb{h})),$$
and 
$$ \lim_{\| \Mb{h} \| \to \infty} \frac{\gamma_W(\Mb{h})}{\ln( \| \Mb{h} \|)} =\infty.$$
Moreover, let $D$ be as in Theorem \ref{Th_R3_Max_Stable}. Let $\{ C(\Mb{x}) \}_{\Mb{x} \in \mathbb{R}^2}= \{ D(Z(\Mb{x})) \}_{\Mb{x} \in \mathbb{R}^2}$. Then, if $\sigma_{C}>0$:
\Ben
\item $\mathcal{R}_2(\cdot, C)$ satisfies the axiom of asymptotic spatial homogeneity of order $-2$ with $K_1(A, C)=0$ and $K_2(A, C)=\sigma_{C}^2/\nu(A), A \in \mathcal{A}_c$.
\item For all $\alpha \in (0,1) \backslash \{ 1/2 \}$, $\mathcal{R}_{3,\alpha}(\cdot, C)$ satisfies the axiom of asymptotic spatial homogeneity of order $-1$ with $K_1(A,C)=\mathbb{E}[C(\Mb{0})]$ and $K_2(A,C)=\sigma_{C} q_{\alpha}/[\nu(A)]^{\frac{1}{2}}$, $A \in \mathcal{A}_c$.
\item For all $\alpha \in (0,1)$, $\mathcal{R}_{4,\alpha}(\cdot, C)$ satisfies the axiom of asymptotic spatial homogeneity of order $-1$ with $K_1(A,C)=\mathbb{E}[C(\Mb{0})]$ and $K_2(A,C)=\sigma_{C} \phi(q_{\alpha})/ \{ [\nu(A)]^{\frac{1}{2}}(1-\alpha) \}$, $A \in \mathcal{A}_c$.
\Een
\end{Th} 

\begin{proof}
The same arguments as in the proof of Theorem \ref{Th_R3_Brown_Resnick_Power_Variogram} show that $C$ satisfies \eqref{Eq_Dependence_Covariance_x-y} and has a constant expectation. As $W$ is sample-continuous, Proposition 13 in \cite{kabluchko2009stationary} gives that $Z$ is sample-continuous. Thus, the same arguments as in the proof of Theorem \ref{Th_R3_Max_Stable} show that $C \in \mathcal{C}$. Moreover, Remark 3 in \cite{koch2017TCL} gives that $C$ satisfies the CLT. Hence, Corollary \ref{Corr_Implication_CLT_ASH_R2} gives the first result. The second result follows from Theorem \ref{Th_Multiple_Results_General_CP}, Point 3. The combination of Proposition \ref{Prop_Condition_Uniform_Integrability} and Point 4 in Theorem \ref{Th_Multiple_Results_General_CP} yields the third result.
\end{proof}

The following result is a direct consequence of Theorem \ref{Th_R3_Brown_Resnick_General_Variogram}.
\begin{Corr}
\label{Corr_R3_Orderminus1_BR_Generalvariogram}
Let $Z$, $D$ and $C$ be as in Theorem \ref{Th_R3_Brown_Resnick_General_Variogram} (but without assuming that $\sigma_C>0$). Moreover, assume that $D$ is non-decreasing and non-constant. 
Then:
\Ben
\item $\mathcal{R}_2(\cdot, C)$ satisfies the axiom of asymptotic spatial homogeneity of order $-2$ with $K_1(A, C)=0$ and $K_2(A, C)=\sigma_{C}^2/\nu(A), A \in \mathcal{A}_c$.
\item For all $\alpha \in (0,1) \backslash \{ 1/2 \}$, $\mathcal{R}_{3,\alpha}(\cdot, C)$ satisfies the axiom of asymptotic spatial homogeneity of order $-1$ with $K_1(A,C)=\mathbb{E}[C(\Mb{0})]$ and $K_2(A,C)=\sigma_{C} q_{\alpha}/[\nu(A)]^{\frac{1}{2}}$, $A \in \mathcal{A}_c$.
\item For all $\alpha \in (0,1)$, $\mathcal{R}_{4,\alpha}(\cdot, C)$ satisfies the axiom of asymptotic spatial homogeneity of order $-1$ with $K_1(A,C)=\mathbb{E}[C(\Mb{0})]$ and $K_2(A,C)=\sigma_{C} \phi(q_{\alpha})/ \{ [\nu(A)]^{\frac{1}{2}}(1-\alpha) \}$, $A \in \mathcal{A}_c$.
\Een
\end{Corr}

\begin{proof}
The random field $Z$ is simple, stationary, sample-continuous (see the proof of Theorem \ref{Th_R3_Brown_Resnick_General_Variogram}) and max-stable. Thus, Proposition 1 in \cite{koch2017TCL} gives that $\sigma_{C}>0$. Consequently, Theorem \ref{Th_R3_Brown_Resnick_General_Variogram} yields the result.
\end{proof}

We conclude this section by commenting on the damage function $D(z) = \mathbb{I}_{ \{ z>u \} }, z>0$, for $u>0$, which is considered in \cite{koch2017spatial}. This function is measurable from $((0, \infty), \mathcal{B}((0,\infty)))$ to $(\mathbb{R},\mathcal{B}(\mathbb{R}))$. Moreover, it is bounded and hence obviously satisfies \eqref{Eq_Assumption_F} for every random field $Z$. Additionally, this function is non-decreasing and non-constant. Consequently, the results of Theorem 3, Point 3 and Theorem 5, Point 2 in \cite{koch2017spatial} concerning the Brown--Resnick random field associated with the variogram $\gamma_W(\Mb{x})=m \| \Mb{x} \|^{\psi}$, where $m >0$ and $\psi \in (0,2]$, and the Smith random field, are particular cases of Corollary \ref{Corr_R3_Orderminus1_BR_Powervariogram}.

\section{Conclusion}
\label{Sec_Conclusion}

In this paper, we first explore the notions of spatial risk measure and corresponding axioms introduced in \cite{koch2017spatial} further as well as describe their utility for both actuarial science and practice. Second, in the case of a general cost field, we provide sufficient conditions such that spatial risk measures associated with expectation, variance, VaR as well as ES and induced by this cost field satisfy the axiom of asymptotic spatial homogeneity of order $0$, $-2$, $-1$ and $-1$, respectively. Finally, in the case where the cost field is a function of a max-stable random field, we give sufficient conditions on both the function and the max-stable field such that spatial risk measures associated with expectation, variance, VaR as well as ES and induced by the resulting cost field satisfy the axiom of asymptotic spatial homogeneity of order $0$, $-2$, $-1$ and $-1$, respectively. Hence, these conditions allow one to know the rate of spatial diversification when the region under study becomes large, which is valuable for the banking/insurance industry. Overall, this paper improves our comprehension of the concept of spatial risk measure as well as of their properties with respect to the space variable and, among others, generalizes several results to be found in \cite{koch2017spatial}. 

Ongoing work consists in the study of concrete examples of spatial risk measures involving max-stable fields and relevant damage functions. Inter alia, we apply our theory to winter storm risk over a specific European region. As previously mentioned, max-stable fields have GEV univariate marginal distributions with three parameters. The first step involves jointly fitting the latter and the dependence parameters of different max-stable models (Smith, Brown--Resnick, \ldots) to wind speed maxima using, e.g., composite likelihood methods \citep[see, e.g.,][]{padoan2010likelihood}. Model selection has then to be performed employing, for instance, the composite likelihood information criterion. The second step consists in choosing an appropriate damage function and exposure field and leads, in combination with the first one, to the cost field model. If the appropriate sufficient conditions mentioned in Section \ref{Subsec_Costfield_Function_Maxstable} are met, then we can draw conclusions about the asymptotic rate of spatial diversification, and less importantly spatial invariance under translation. Simulating from the cost field, we obtain realizations of the normalized spatially aggregated loss on chosen sub-regions, which allow the estimation of the spatial risk measures of interest. This enables one to check, e.g., whether the axiom of spatial sub-additivity is satisfied.

Future work will include the study of spatial risk measures associated with other classical risk measures (e.g., more general distorsion risk measures than VaR or ES and expectile risk measures) and/or induced by cost fields involving other kinds of random fields than max-stable fields. For instance, it would be worthwhile to investigate whether spatial risk measures associated with VaR and ES can still satisfy the axiom of asymptotic spatial homogeneity of order $-1$ in the case where the cost field does not satisfy the CLT.

\section*{Acknowledgements}

The author would like to thank Anthony C. Davison and Christian Y. Robert for some interesting comments. He also acknowledges Paul Embrechts and Ruodu Wang for providing some references about robustness of risk measures as well as Gennady Samorodnitsky for a fruitful exchange about integrals of random fields. Finally, he is also grateful to the Associate Editor and two anonymous referees for insightful suggestions. This research was partly funded by the Swiss National Science Foundation grant number 200021\_178824.

\newpage
\bibliographystyle{apalike}
\bibliography{References_Erwan}

\end{document}